\documentclass[a4paper,UKenglish,runningheads,11pt]{llncs}

\usepackage{tabularx,booktabs,multirow,delarray,array}
\usepackage{graphicx,amssymb,amsmath}
\usepackage{enumerate}
\usepackage[ruled,vlined,linesnumbered]{algorithm2e}
\usepackage{fullpage}
\usepackage{latexsym}
\usepackage{enumerate}
\usepackage{lineno}
\usepackage{wrapfig}
\usepackage{bibspacing}
\newenvironment{proof}{\par\noindent{\bf Proof:}}{\mbox{}\hfill$\qed$\\}
\newtheorem{observation}{Observation}

\newcommand{\ignore}[1]{ }

\newcounter{rem}
\def\qed{\hbox{\rlap{$\sqcap$}$\sqcup$}}

\begin{document}

\title{Visibility Polygons and Visibility Graphs among Dynamic Polygonal Obstacles in the Plane}
\titlerunning{Visibility polygons and visibilty graphs among dynamic obstacles}

\author{
Sanjana Agrawal\inst{1}
\and
R. Inkulu\inst{1}
\thanks{A preliminary version \cite{conf/cocoon/AgrawInkulu20} appeared in Proceedings of the $26^{th}$ International Computing and Combinatorics Conference (COCOON), 2020.}
}

\institute{
Department of Computer Science \& Engineering\\
IIT Guwahati, India\\
\email{\{sanjana18,rinkulu\}@iitg.ac.in}
}

\maketitle

\pagenumbering{arabic}
\setcounter{page}{1}

\begin{abstract}
We devise an algorithm for maintaining the visibility polygon of any query point in a dynamic polygonal domain, i.e., as the polygonal domain is modified with vertex insertions and deletions to its obstacles, we update the data structures that store the visibility polygon of the query point.
After preprocessing the initial input polygonal domain to build a few data structures, our algorithm takes $O(k(\lg{|VP_{\cal P'}(q)|})+(\lg{n'})^{2}+h)$ (resp. $O(k(\lg n')^2+(\lg|VP_{\cal P'}(q)|)+h)$) worst-case time to update data structures that store visibility polygon $VP_{\cal P'}(q)$ of a query point $q$ when any vertex $v$ is inserted to (resp. deleted from) any obstacle of the current polygonal domain $\cal P'$.
Here, $n'$ is the number of vertices in $\cal P'$, $h$ is the number of obstacles in $\cal P'$, $VP_{\cal P'}(q)$ is the visibility polygon of $q$ in $\cal P'$ ($|VP_{\cal P'}(q)|$ is the number of vertices of $VP_{\cal P'}(q)$), and $k$ is the number of combinatorial changes in $VP_{\cal P'}(q)$ due to the insertion (resp. deletion) of $v$.
\ignore {
Our algorithm takes $O(k(\lg{|VP_{\cal P'}(q)|})+(\lg{n'})^{2}+h)$ worst-case time to update the visibility polygon $VP_{\cal P'}(q)$ of a query point $q$ when any vertex $v$ is inserted to any obstacle of the current polygonal domain $\cal P'$, and it takes $O(k(\lg n')^2+(\lg|VP_{\cal P'}(q)|)+h)$ worst-case time to update $VP_{\cal P'}(q)$ when any vertex $v$ of any obstacle of the current polygonal domain $\cal P'$ is deleted.
In $O(|VP_{\cal P'}(q)|(\lg{n'})^2+h)$ time, apart from outputting the visibility polygon of a given query point, our visibility polygon query algorithm builds data structures that facilitate in efficiently updating the visibility polygon when a vertex is inserted to or deleted from any of the obstacles.
Initially, we preprocess the input polygonal domain $\cal P$ in $O(n(\lg{n})^2+h(\lg{h})^{1+\epsilon})$ time, and build data structures of size $O(n)$.
Here, $n$ is the number of vertices of $\cal P$, $h$ is the number of obstacles in $\cal P$, and $\epsilon > 0$ is a small positive constant (resulting from triangulating the free space of $\cal P$ using the algorithm in \cite{journals/ijcga/Bar-YehudaC94}).
}
\hfil\break

As an application of the above algorithm, we also devise an algorithm for maintaining the visibility graph of a dynamic polygonal domain, i.e., as the polygonal domain is modified with vertex insertions and deletions to its obstacles, we update data structures that store the visibility graph of the polygonal domain.
After preprocessing the initial input polygonal domain, our dynamic algorithm takes $O(k(\lg{n'})^{2}+h)$ (resp. $O(k(\lg{n'})^{2}+h)$) worst-case time to update data structures that store the visibility graph when any vertex $v$ is inserted to (resp. deleted from) any obstacle of the current polygonal domain $\cal P'$.
Here, $n'$ is the number of vertices in $\cal P'$, $h$ is the number of obstacles in $\cal P'$, and $k$ is the number of combinatorial changes in the visibility graph of $\cal P'$ due to the insertion (resp. deletion) of $v$.
\end{abstract}

\section{Introduction}
\label{sect:intro}

The polygonal domain comprises a set of pairwise-disjoint simple polygons (obstacles) in the plane.
We assume the obstacles in the polygonal domain are placed in a large bounding box.
For any polygonal domain $\cal P$, the {\it free space} $\mathcal{F(P)}$ is the closure of bounding box without the union of the interior of all the obstacles in $\cal{P}$.
Any two points $p, q \in {\cal{F(P)}}$ are {\it visible} to each other if the open line segment joining $p$ and $q$ lies entirely in $\cal{F(P)}$.
A vertex $v$ of the polygonal domain is said to be a {\it visible vertex} to a point $q$ whenever $v$ is visible to $q$.
For a point $q \in {\cal{F(P)}}$, the {\it visibility polygon} of $q$, denoted by $VP_{\cal{P}}(q)$, is the maximal set $S$ of points in $\cal{F(P)}$, such that each point in $S$ is visible to $q$.
(When $\cal P$ is clear from the context, the visibility polygon of $q$ is denoted by $VP(q)$.)
The {\it visibility polygon query} problem seeks to preprocess the given polygonal domain $\cal P$ so that to efficiently compute $VP_{\cal P}(q)$ for any query point $q$ located in $\cal{F(P)}$.
Computing visibility polygons is a fundamental problem in computational geometry, and it is studied extensively.
The {\it (vertex-vertex) visibility graph} of a polygonal domain $\cal P$, denoted by $VG_{\cal P}$, is the undirected graph with its vertex set comprising all the vertices of $\cal P$. 
The edge set of $VG_{\cal P}$ comprises of every line segement with its endpoints $v', v''$ being the vertices of $\cal P$ such that $v'$ and $v''$ are visible to each other among the obstacles of $\cal P'$. 
The visibility graphs have many applications, ex., in computing geodesic shortest paths, minimum link paths, and in computing the weak visibility polygons.
In the following, we denote the number of vertices of the simple polygon or the polygonal domain by $n$, the number of obstacles in the polygonal domain by $h$, and the set comprising the edges of the visibility graph by $E$. 

The problem of computing the visibility polygon of a point in a simple polygon was first attempted by Davis and Benedikt in \cite{journals/cgip/Davis79}, and they presented an $O(n^2)$ time algorithm.
Later, both ElGindy and Avis \cite{journals/jal/ElGindyA81}, and Lee \cite{journals/cvgip/Lee83}, presented $O(n)$ time algorithms for the same problem.
Joe and Simpson \cite{journals/bitnummath/joe87} corrected a flaw in \cite{journals/jal/ElGindyA81,journals/cvgip/Lee83}, and devised an $O(n)$ time algorithm that correctly handles winding in the simple polygon. 
For a simple polygon with holes, both Suri and O' Rourke \cite{conf/compgeom/SuriO86}, and Asano \cite{journals/algorithmica/AsanoAGHI86} presented $O(n\lg{n})$ time algorithms, and Heffernan and Mitchell \cite{journals/siamcomp/HeffernanM95} gave a $O(n+h\lg{h})$ time algorithm.
The visibility polygon computation among convex sets was considered by Ghosh in \cite{journals/jal/Ghosh91}.

For both the simple polygon as well as the polygonal domain, the previous works considered the visibility polygon query problem.
Bose et~al. \cite{journals/comgeo/BoseLM02} gave an algorithm to preprocess the given simple polygon in $O(n^3 \lg{n})$ time, build data structures of size $O(n^3)$, and answer any visibility polygon query in $O(\lg{n}+|VP(q)|)$ time.
Later, Aronov et~al. \cite{journals/dcg/AronovGTZ02} devised an algorithm for the same problem with preprocessing time $O(n^2\lg{n})$, space $O(n^2)$, and query time $O(\lg^2{n}+|VP(q)|)$.
Zarei and Ghodsi \cite{conf/compgeom/ZareiG05} presented an algorithm that preprocesses the given polygonal domain in $O(n^3 \lg{n})$ time to build data structures of size $O(n^3)$, and answers each visibility polygon query in $O((1+h')\lg{n}+|VP(q)|)$ time, where $h' = min(h, |VP(q)|)$.
The algorithm by Inkulu and Kapoor \cite{journals/comgeo/InkuluK09} preprocesses the input polygonal domain in $O(n^2\lg{n})$ time, builds data structures of size $O(n^2)$, and answers visibility polygon query in $O(\min(h, |VP(q)|)(\lg{n})^2 + h + |VP(q)|)$ time. 
This paper also presented another algorithm with preprocessing time $O(T + |VG| + n\lg{n})$, space $O(\min(|VG|, h n) + n)$, and query time $O(|VP(q)|\lg{n} + h)$.
Here, $|VG|$ denotes the number of edges in the visibility graph, and $T$ is the time to triangulate the free space of the given polygonal domain.
Baygi and Ghodsi \cite{journals/cgta/BaygiG13} constructed a data structure of size $O(n^2)$ in $O(n^2\lg{n})$ time, and their algorithm answers any visibility polygon query in $O(|VP(q)| + \lg{n})$ time.
Lu et~al. \cite{journals/jics/LuYW11} presented an algorithm to compute a data structure of size $O(n^2)$ in $O(n^2\lg{n})$ time, which helps in answering any visibility polygon query in $O(|VP(q)|+(\lg{n})^2+h\lg(n/h))$ time.
Chen and Wang \cite{journals/comgeo/ChenW15} gave an algorithm that preprocesses the polygonal domain in $O(n + h^2\lg{h})$ time to construct data structures of size $O(n+h^2)$, so that to answer any visibility polygon query in $O(|VP(q)|\lg{n})$ time.
Pocchiola and Vegter \cite{journals/dcg/Pocch96} considered the query version of visibility polygon computation in the polygonal domain comprising convex obstacles.
Their algorithm computes the visibility polygon of any query point in $O(|VP(q)| \lg{n})$ time after preprocessing the convex polygonal domain in $O(n\lg{n})$ time and building data structures of size $O(n)$.  

To compute the visibility graph of a simple polygon, Lee~\cite{phdthes/Lee78} and Sharir and Schorr~\cite{journals/siamcomp/SharirS86} gave algorithms that take $O(n^2\lg{n})$ time.
For this problem, algorithms given by Asano et~al.~\cite{journals/algorithmica/AsanoAGHI86} and Welzl~\cite{journals/ipl/Welzl85} take $O(n^2)$ worst-case time, which are worst-case optimal since the number of edges $|E|$ in the visibility graph is $\Theta(n^2)$ in the worst-case.
For triangulated simple polygons, Hershberger~\cite{journals/algorithmica/Hersh89} gave an $O(n+|E|)$ time algorithm. 
Since any simple polygon can be triangulated in linear time due to an algorithm by Chazelle~\cite{journals/dcg/Chazelle91}, the algorithm in \cite{journals/algorithmica/Hersh89} essentially takes $O(n+|E|)$ time, which is optimal for this problem.

There are a number of algorithms to compute the visibility graph of a polygonal domain efficiently.
Overmars and Welzl~\cite{conf/socg/OverWelzl88} gave an output-sensitive algorithm that takes $O(|E|\lg{n})$ time and $O(n+|E|)$ space.
Ghosh and Mount~\cite{journals/siamcomp/GhoshM91} presented an algorithm with $O(n\lg{n} + |E|)$ worst-case time with $O(|E|+n)$ space.
Keeping the same time complexity as the algorithm in \cite{journals/siamcomp/GhoshM91}, Pocchiola and Vegter \cite{journals/dcg/Pocch96} improved the space complexity to $O(n)$.
Kapoor and Maheshwari~\cite{conf/socg/Kapoor88,journals/siamcomp/KapoorM00} proposed an algorithm with time complexity $O(T + |E| + h\lg{n})$ and $O(|E|)$ space.
Here, $T$ is the time for triangulating the free space of the polygonal domain, which is $O(n+h(\lg{h})^{1+\delta})$ due to Bar-Yehuda and Chazelle~\cite{journals/ijcga/Bar-YehudaC94} for a small positive constant $\delta$.
Later, closely following \cite{journals/siamcomp/GhoshM91}, Chen and Wang~\cite{journals/jocg/ChenW15b} devised an algorithm with the same time complexity as the algorithm in \cite{conf/socg/Kapoor88,journals/siamcomp/KapoorM00}.
The texts by Ghosh~\cite{books/visalgo/skghosh2007} and O' Rourke~\cite{books/artgal/rourke1987} detail a number of algorithms for computing visibility polygons and visibility graphs in the plane.

In the context of visibility polygons (resp. visibility graphs), having algorithms to maintain the visibility polygon (resp. visibility graph) among dynamic polygonal obstacles helps in updating the visibility polygon (resp. visibility graph) efficiently (that is, with respect to updation time) as compared to computing the entire visibility polygon (resp. visibility graph) from scratch using traditional algorithms.
In doing this, the algorithm specifically exploits the recent changes that occurred to the polygonal domain; based on these changes, the update is performed. 
%A dynamic algorithm is said to be {\it incremental} if it updates the visibility polygon of the given point efficiently whenever a new vertex is inserted to any of the obstacles of the current polygonal domain.
%Similarly, an algorithm is {\it decremental} if it efficiently updates the visibility polygon of the given point whenever a vertex of any of the obstacles of the current polygonal domain is deleted. 
%If the dynamic algorithm is both incremental as well as decremental, then it is termed a {\it fully dynamic} algorithm.
Both Inkulu and Nitish~\cite{conf/caldam/Inkulu17} and Inkulu et al.~\cite{journals/ijcga/InkuluST20} devised algorithms for maintaining the visibility polygon of a query point in a dynamic simple polygon.
After preprocessing the initial simple polygon with $n$ vertices in $O(n)$ time, the algorithm given in \cite{journals/ijcga/InkuluST20} updates data structures that store the visibility polygon of a query point in $O((k+1)(\lg{n'})^2)$ time when any vertex $v$ is inserted to (resp. deleted from) the simple polygon.
Here, $k$ is the number of combinatorial changes in the visibility polygon of $q$ due to the insertion (resp. deletion) of $v$ and $n'$ is the number of vertices in the current simple polygon.
An algorithm for maintaining the weak visibility polygon in a dynamic simple polygon is given in \cite{journals/ijcga/InkuluST20}.
Choudhury and Inkulu \cite{conf/caldam/Inkulu19} devised an algorithm for maintaining the visibility graph of a dynamic simple polygon. 
The algorithm in \cite{conf/caldam/Inkulu19} preprocesses the initial simple polygon in $O(n + |E| \lg{|E|})$ time to build data structures of size $O(n+|E|)$, and updates the visibility graph of the current simple polygon $\cal P'$ in $O((k+1)(\lg{n'})^2)$ time when a vertex $v$ is inserted to (resp. deleted from) $\cal P'$. 
Here, $k$ is the number of combinational changes in the visibility graph of $\cal P'$ due to the insertion (resp. deletion) of $v$ to $\cal P'$.
The visibility in the context of a moving observer was studied in \cite{journals/ipl/ChenDaescu98,conf/cccg/AkbariGhodsi10}.

\subsection{Our results}

In this paper, two algorithms are proposed. 
After preprocessing to build a few data structures, the first algorithm updates data structures that store the visibility polygon of any query point $q$, and the second algorithm updates the data structures that store the visibility graph of the current polygonal domain, both among a set of dynamic polygonal obstacles. 
These algorithms update relevant data structures whenever any vertex is inserted to any of the obstacles in the current polygonal domain or whenever any vertex is deleted from any of the obstacles.

\begin{itemize}
\item
Let $VP_{\cal P'}(q)$ be the visibility polygon of a query point $q$ located in the free space of the current polygonal domain $\cal P'$.
Also, let $n'$ be the number of vertices of $\cal P'$.
When any vertex $v$ is inserted to (resp. deleted from) any of the obstacles of $\cal P'$, our algorithm takes $O(k(\lg{|VP_{\cal P'}(q)|})+(\lg{n'})^{2}+h)$ (resp. $O(k(\lg n')^2+(\lg|VP_{\cal P'}(q)|)+h)$) time to update data structures that store $VP_{\cal P'}(q)$. 
Here, $k$ is the number of combinatorial changes in $VP_{\cal P'}(q)$ due to the insertion (resp. deletion) of $v$.
(The combinatorial changes to $VP_{\cal P'}(q)$ include vertices inserted, vertices deleted, edges inserted, and edges deleted from $VP_{\cal P'}(q)$.)
Given any query point $q$ located in the free space of the current polygonal domain $\cal P'$, our output-sensitive visibility polygon query algorithm computes $VP_{\cal P'}(q)$ in $O(|VP_{\cal P'}(q)|(\lg{n'})^2+h)$ time. 
The data structures constructed as part of answering the visibility polygon query algorithm facilitate efficient updations as the polygonal domain changes. 
We preprocess the initial input polygonal domain $\cal P$ in $O(n(\lg{n})^2+h(\lg{h})^{1+\epsilon})$ time, and construct data structures of size $O(n)$.
Here, $n$ is the number of vertices of $\cal P$, $h$ is the number of polygonal obstacles in $\cal P$, and $\epsilon > 0$ is a small positive constant (resulting from triangulating the free space of $\cal P$ using the algorithm in \cite{journals/ijcga/Bar-YehudaC94}). 

\vspace{0.12in}
\item
Let $VG_{\cal P'}$ be the visibility graph of the current polygonal domain $\cal P'$.
Also, let $n'$ be the number of vertices in $\cal P'$.
When any vertex $v$ is inserted to (resp. deleted from) any of the obstacles of $\cal P'$, our algorithm takes $O(k(\lg{n'})^2 + h)$ (resp. $O(k((\lg n')^2+h))$) time to update the visibility graph of $\cal P'$. 
Here, $k$ is the number of combinatorial changes in $VG_{\cal P'}$ due to the insertion of $v$.
(The combinatorial changes in $VG_{\cal P'}$ include vertices inserted, vertices deleted, edges inserted, and edges deleted from $VG_{\cal P'}$.)
We preprocess the initial input polygonal domain $\cal P$ in $O(n(\lg{n})^2+h(\lg{h})^{1+\epsilon} + |E| \lg{|E|})$ time, and construct data structures of size $O(n + |E|)$.
Here, $n$ is the number of vertices of $\cal P$, $h$ is the number of polygonal obstacles in $\cal P$, $E$ is the set of edges in the visibility graph of $\cal P$, and $\epsilon > 0$ is a small positive constant (resulting from triangulating the free space of $\cal P$ using the algorithm in \cite{journals/ijcga/Bar-YehudaC94}). 
\end{itemize}

For convenience, we assume all through the algorithm, the number of polygonal obstacles does not change, and the obstacles are pairwise disjoint.
To our knowledge, these are the first algorithms for maintaining the visibility polygon of any given point and the visibility graph of the polygonal domain among the dynamic polygonal obstacles. 
These algorithms obviate computing the visibility polygon of a given point or the visibility graph from scratch whenever the polygonal domain is modified with a vertex insertion or vertex deletion. 
 
Next, we give an overview of our approach.
For any given query point $q$, we compute two binary trees (called visibility trees) to store the visibility polygon of $q$.
These visibility tree data structures in our algorithm are a modification to visibility trees defined in \cite{journals/siamcomp/KapoorM00}.  
From these trees, we compute the visibility polygon $VP_{\cal P'}(q)$ of $q$ in the current polygonal domain $\cal P'$.
Whenever a new vertex $v$ is inserted to an obstacle of the current polygonal domain $\cal P'$, first, we perform a ray-shooting query to determine whether $v$ is visible to $q$.
If it is not visible, then our algorithm does nothing further. 
Otherwise, we traverse the visibility trees in depth-first order to determine all the vertices that are not visible to $q$ due to the insertion of $v$. 
We update the $VP_{\cal P'}(q)$ by removing these vertices from $VP_{\cal P'}(q)$.
As part of this, we also update both the visibility trees.
Similarly, in the case of deletion of a vertex $v$, by performing a ray-shooting query, it is determined whether deleted vertex $v$ is visible to $q$.
If the deleted vertex $v$ was visible to query point $q$, then the deletion of $v$ may cause the addition of some new vertices to $VP_{\cal P'}(q)$.
These new vertices are determined using our output-sensitive visibility polygon query algorithm with $q$ as the query point.
Updating $VP_{\cal P'}(q)$ involves merging a set of polygons due to these vertices into data structures that store $VP_{\cal P'}(q)$.
Our visibility polygon query algorithm enhances data structures designed in \cite{journals/comgeo/InkuluK09}, so that they work efficiently among dynamic polygonal obstacles.
As in \cite{journals/comgeo/InkuluK09}, our visibility polygon query algorithm computes visibility trees of a query point $q$ by determining sequences' of traingles bounding rays of visibility cones initiated at $q$ intersect. 
When any ray in a visibility cone is found to strike a point on the boundary of any obstacle $O$, we compute tangents from $q$ to $O$.
To compute these tangents and, in turn, for updating visibility cones, our algorithm dynamically maintains hull trees of obstacle boundaries using the algorithm by Overmars and van Leeuwen \cite{books/compgeom/prep1985}.
In updating the visibility graph, we use an important characterization to determine edges of the visibility graph $VG_{\cal P'}$ of the polygonal domain $\cal P'$ that need to be deleted from (resp. included into) $VG_{\cal P'}$ due to the insertion (resp. deletion) of a vertex from $\cal P'$. 

\subsection{Terminology}

We denote the initial input polygonal domain by $\cal P$, the number of vertices of $\cal P$ by $n$, and the number of polygonal obstacles in $\cal P$ by $h$.
We use $\cal P'$ to denote a polygonal domain just before inserting (resp. deleting) a vertex, and $\cal P''$ to denote a polygonal domain just after inserting (resp. deleting) a vertex.
We call $\cal P'$ as the {\it current polygonal domain}, and we call $\cal P''$ as the {\it updated polygonal domain}.
The number of vertices of $\cal P'$ is denoted by $n'$.
The number of pairwise disjoint polygonal obstacles in any polygonal domain is denoted by $h$.
Recall the definitions of free space of a polygonal domain, visibility polygon of a point in a polygonal domain, and the visibility graph of a polygonal domain from Section~\ref{sect:intro}. 
It is assumed that every new vertex is added between two successive vertices of an obstacle. 
Whenever a new vertex $v$ is inserted between two adjacent vertices $v_i$ and $v_{i+1}$ in polygonal domain $\cal P'$, it is assumed that two new edges are added to $\cal P'$: one between vertices $v$ and $v_{i}$, and the other between vertices $v$ and $v_{i+1}$. 
Similarly, in the case of deletion of a vertex $v$, it is assumed that after deleting $v$ from $\cal P'$ which is adjacent to vertices $v_i$ and $v_{i+1}$ in $\cal P'$, a new edge is inserted between vertices $v_i$ and $v_{i+1}$.
After adding (resp. deleting) any vertex to (resp. from) any obstacle $O$ in the current polygonal domain, our algorithm assumes $O$ continues to be a simple polygon.
Further, the newly inserted vertex in $\cal P''$ is assumed to be contained in the bounding box containing the obstacles of $\cal P$.
The visibility polygon of $q$ in any polygonal domain $\cal Q$ is denoted by $VP_{\cal Q}(q)$.
The visibility graph in any polygonal domain $\cal Q$ is denoted by $VG_{\cal Q}$.
For any simple polygon $P$, the boundary of $P$ is denoted by $bd(P)$.
Unless specified otherwise, the boundary of any simple polygon is assumed to be traversed in counter-clockwise order.
It is assumed that all the angles are measured in the counter-clockwise direction from the positive horizontal axis ($x$-axis).
A {\it constructed edge} \cite{books/visalgo/skghosh2007} is an edge $u_iu_{i+1}$ on the boundary of $VP_{\cal P'}(q)$ such that either (i) no point of $u_iu_{i+1}$, except the points $u_i$ and $u_{i+1}$, belongs to the boundary of $\cal P'$, and at least one of $u_i,u_{i+1}$ is a vertex of $\cal P'$, or (ii) every point on the edge $u_iu_{i+1}$ is lying on the boundary of $\cal P'$, and neither $u_i$ and $u_{i+1}$ is a vertex of $\cal P'$.
For every constructed edge $u_iu_{i+1}$, among $u_i$ and $u_{i+1}$, the farthest one from $q$ is termed a {\it constructed vertex} of $VP(q)$.

\begin{figure}[ht]
\begin{center}
\includegraphics[totalheight=1.4in]{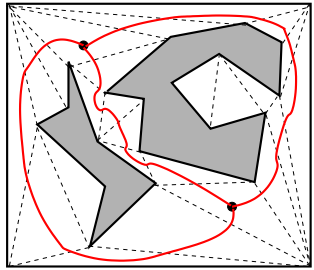}
\caption{\footnotesize Illustrating a triangulation of the free
space among two obstacles and the corridors (indicated by red solid curves).
There are two junction triangles marked by a large dot inside
each of them, connected by three solid (red) curves. Removing the two
junction triangles results in three corridors.
(This illustration is from \cite{journals/jocg/ChenIW16}.)
}
\label{fig:corridors}
\end{center}
\end{figure}
Next we describe a few structures from \cite{conf/stoc/Kapoor99,conf/socg/Kapoor88,journals/dcg/KapoorMM97} which result from decomposing the free space $\cal{F(P)}$ of a polygonal domain $\cal P$.
Consider a triangulation $\cal{T_P}$ of $\cal{F(P)}$. 
Let $G_d$ be the dual graph of $\cal{T_P}$.
First, we prune $G_d$ by iteratively removing all the vertices of degree one.
Let $G'_d$ be the pruned dual graph of $\cal{T_P}$.
For every node belonging to $G'_d$, its degree is either two or three.
Every triangle in $\cal{T_P}$ that corresponds to a vertex of degree three in $G'_d$ is known as a {\it junction}.
Removing the junctions leads to connected simple polygonal regions in $\cal P$.
Every such region is known as a {\it corridor}.
(Refer to Fig.~\ref{fig:corridors}.)
Any corridor $C$ has at most two poly-lines, each is known as a {\it corridor side}, and two edges, each is known as a {\it (corridor) bounding edge} of $C$.
For any corridor $C$, and for any bounding edge $v'v''$ of $C$, if the angle made by ray $qv'$ is less than or equal to (resp. greater than) the angle made by ray $qv''$ at $q$, then the corridor side of $C$ that has $v'$ is known as the {\it left side} (resp. {\it right side}) of $C$ (with respect to $q$).

Let $r'$ and $r''$ be two rays with origin at $p$.
Let $uv_1$ and $uv_2$ be the unit vectors along the rays $r'$ and $r''$, respectively.
A {\it cone} $C_p(r', r'')$ is the set of points defined by rays $r'$ and $r''$ such that a point $p' \in C_p(r', r'')$ if and only if $p'$ can be expressed as a convex combination of the vectors $uv_1$ and $uv_2$ with positive coefficients.
When the rays are evident from the context, we denote the cone with $C_p$.
A cone $C_p$ is called a {\it visibility cone} whenever $C_p$ contains at least one point in $\cal{F(P)}$ that is visible to $p$. 
For any cone $C_p(r_i, r_j)$, among rays $r_i$ and $r_j$, the ray that makes lesser angle with the positive $x$-axis at $p$ is the {\it left ray} of $C_p$ and the other ray is the {\it right ray} of $C_p$.
For any cone, throughout the paper, we assume the counter-clockwise angle between the left ray of $C_p$ and the right ray of $C_p$ is less than $\pi$. 
Any edge in the visibility graph of the polygonal domain is called a {\it visible edge}.
We use {\it vertex} to denote any endpoint of any edge of the polygonal domain, and we use {\it node} to denote any tree node in data structures that we construct.

The preprocessing algorithm and the data structures for maintaining the visibility polygon are detailed in Section~\ref{sect:preprocessds}.
Section~\ref{sect:insertdelete} details the algorithms to maintain the visibility polygon of any query point whenever the current polygonal domain is modified with a vertex insertion or with a vertex deletions. 
The output-sensitive visibility polygon query algorithm of any given query point is given in Section~\ref{sect:queryalgo}.
Section~\ref{sect:visgrmaint} details both the preprocessing algorithm and the algorithms for maintaining the visibility graph.
Conclusions are in Section~\ref{sect:conclu}.

\section{Preprocessing algorithm and data structures} 
\label{sect:preprocessds}

We first preprocess the input polygonal domain $\cal P$.
Using the algorithms in \cite{conf/stoc/Kapoor99,conf/socg/Kapoor88,journals/dcg/KapoorMM97}, we partition the free space $\cal{F(P)}$ of $\cal P$ into $O(h)$ corridors and junctions.
To efficiently compute tangents to dynamic sides of corridors, for each side $S$ of every corridor, we construct a hull tree corresponding to $S$ with the algorithm in \cite{journals/jcss/Overvan81,books/compgeom/prep1985}.
For locating any vertex that is inserted to (resp. deleted from) any obstacle, and for locating query points, we compute a point location data structure for the corridor structure with the algorithm in \cite{journals/jal/GoodrichT97}.

\begin{figure}[ht]
\centering
\minipage{0.3\textwidth}
\includegraphics[width=5cm]{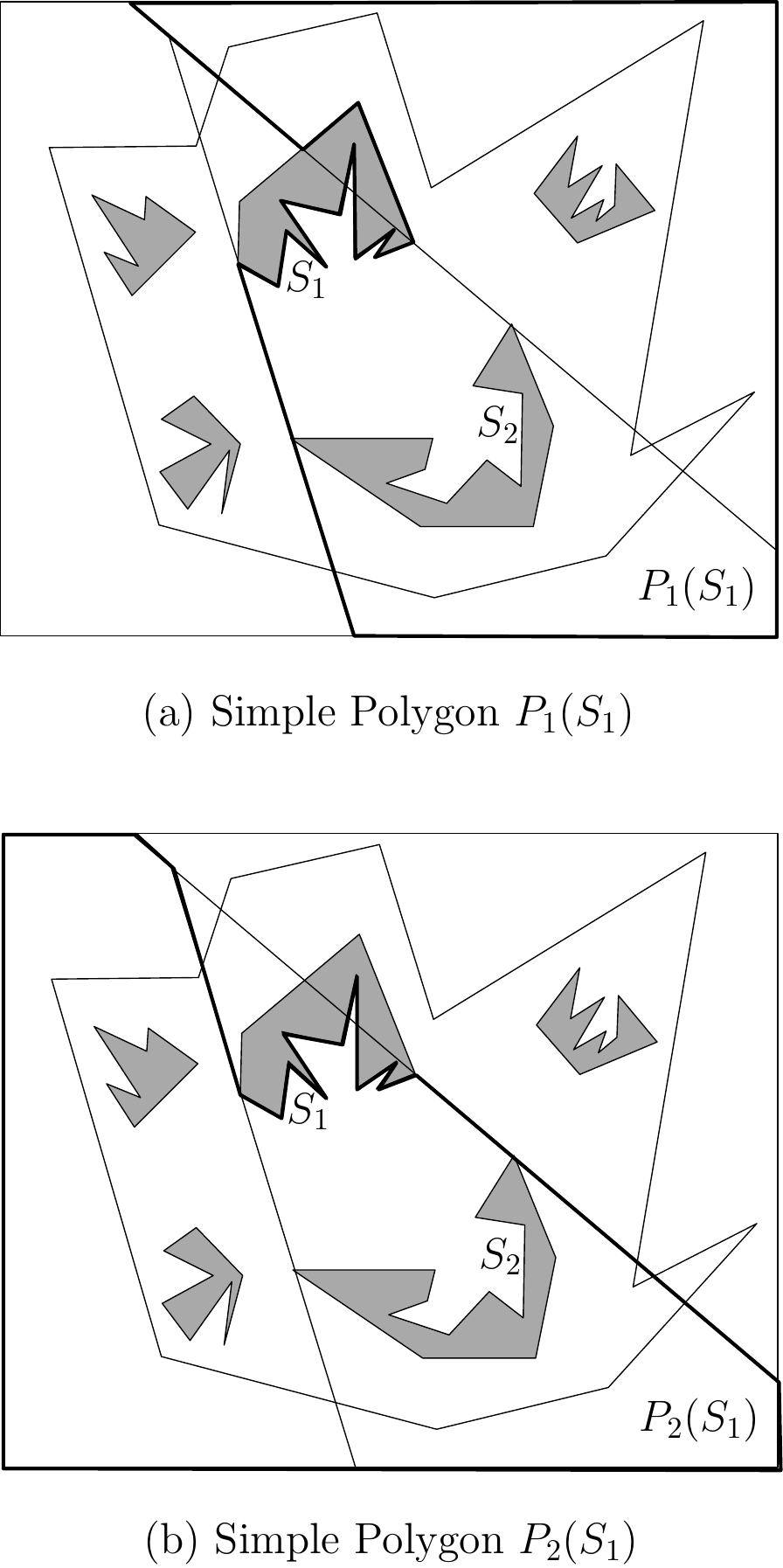}
\endminipage\hfil
\minipage{0.3\textwidth}
\includegraphics[width=5cm]{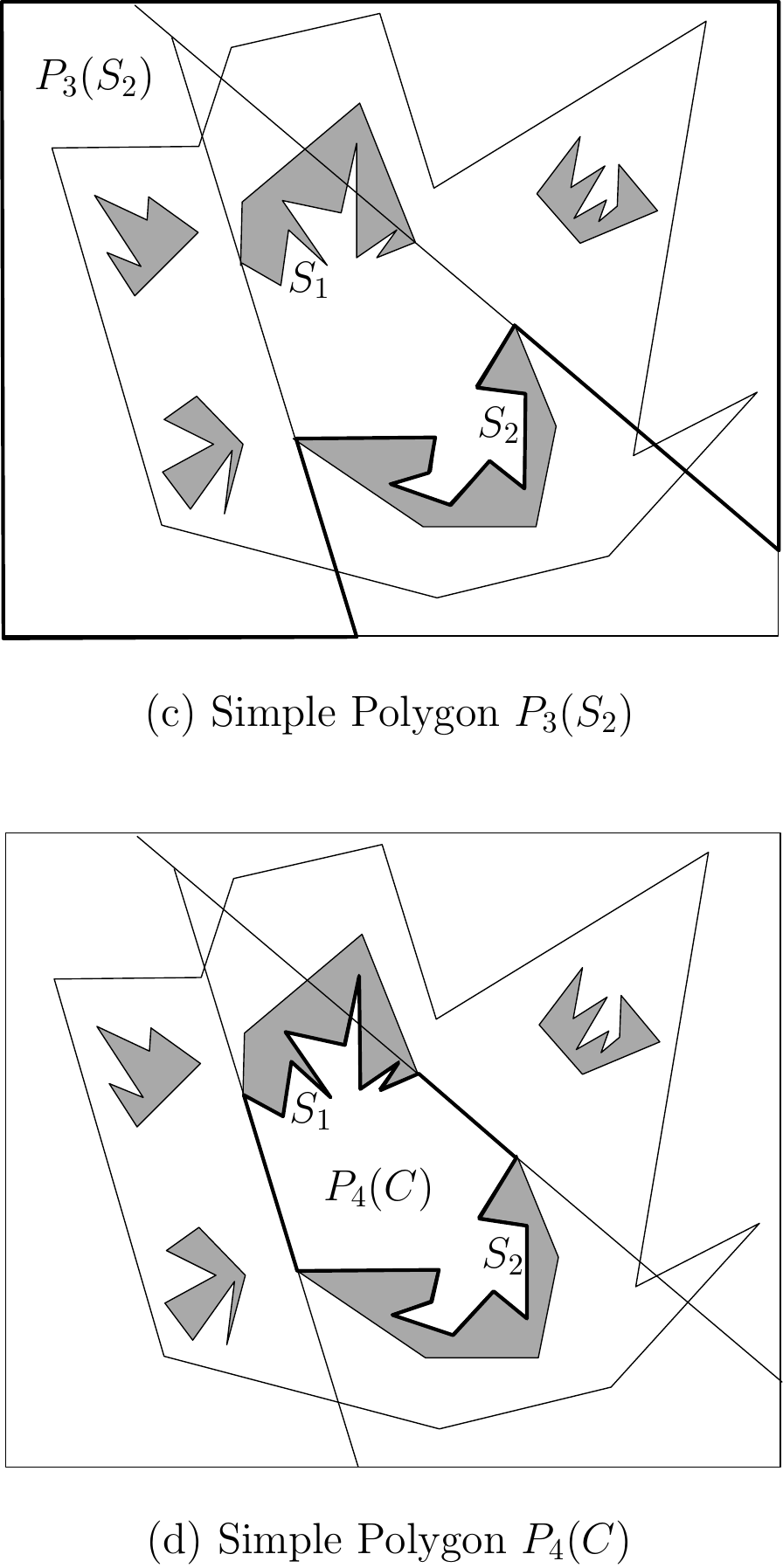}
\endminipage
\caption{Illustrating four simple polygons in $P_C$ corresponding to a corridor $C$.  (This illustration is from \cite{journals/comgeo/InkuluK09}.)}
\label{fig:foursimppoly}
\end{figure}

As in the algorithm in \cite{journals/comgeo/InkuluK09}, for every corridor $C$, we construct a set $P_C$ of simple polygons corresponding to corridor $C$.
(Refer to Fig.~\ref{fig:foursimppoly}.)
One of the simple polygons in $P_C$, denoted with $P_4(C)$, is the corridor $C$ itself.
This polygon helps in determining vertices of $C$ that are visible to $q$ when $q$ is located in $C$.
If $q$ is not located in $C$, the other three simple polygons in $P_C$, denoted by $P_1(S_1), P_2(S_1), P_3(S_2)$, each corresponding to a side of $C$, together help in determining vertices of $C$ that are visible to $q$.
In specific, two of these simple polygons $P_1(S_1), P_2(S_1)$ correspond to one side $S_1$ of $C$, and $P_3(S_2)$ correspond to the other side $S_2$ of $C$.
For more details, refer to \cite{journals/comgeo/InkuluK09}. 

Whenever an obstacle in the polygonal domain gets modified with a vertex $v$ insertion or deletion, a side of the corridor on which $v$ resides is updated.
Whenever a side $S$ of any corridor $C$ changes, we make the changes to $P_4(C)$ as well as to the simple polygons in $P_C$ that correspond to side $S$. 
In the query phase, using the algorithm in \cite{journals/ijcga/InkuluST20}, we compute visible vertices and constructed vertices (refer to \cite{books/visalgo/skghosh2007}) in these dynamic simple polygons.
As part of preprocessing required for the algorithm in \cite{journals/ijcga/InkuluST20}, for every corridor $C$, each of the four simple polygons in $P_C$ are further processed in linear time to construct data structures as required by the algorithm in \cite{journals/ijcga/InkuluST20}.
Analogous to set $P_C$ of simple polygons for $\cal P$, our algorithm maintains corresponding set $P'_C$ (resp. $P''_C$) of simple polygons for the current polygonal domain $\cal P'$ (resp. the updated polygonal domain $\cal P''$).

\begin{lemma} 
\label{lem:vispolypreproc}
Given a polygonal domain $\cal P$ defined with $h$ obstacles and $n$ vertices, our preprocessing algorithm computes data structures of size $O(n)$ in $O(n(\lg{n})^2 + h(\lg{h})^{1+\epsilon})$ time.  
Here, $\epsilon$ is a small positive constant resulting from the triangulation of free space $\cal{F(P)}$ using the algorithm in \cite{journals/ijcga/Bar-YehudaC94}.
\end{lemma}
\begin{proof}
Triangulating $\cal P$ using the algorithm in \cite{journals/ijcga/Bar-YehudaC94} and computing corridors using the algorithms in \cite{conf/stoc/Kapoor99,conf/socg/Kapoor88,journals/dcg/KapoorMM97} takes $O(n+h(\lg{h})^{1+\epsilon})$ time. 
Construction of point location data structure requires $O(n)$ time with the algorithm in \cite{journals/jal/GoodrichT97}.
Using the algorithm in \cite{books/compgeom/prep1985}, constructing hull trees for all the corridor sides together takes $O(n(\lg n)^2)$ time.
Computing the set $P_C$ of polygons \cite{journals/comgeo/InkuluK09} for every corridor $C$, while considering all the corridors together takes $O(n)$ time.
And, preprocessing all such simple polygons corresponding to all the corridors for constructing data structures required for the algorithm in \cite{journals/ijcga/InkuluST20} takes $O(n)$ time.
The data structures for hull trees and point location data structures are of size $O(n)$.
All the simple polygons constructed for all the corridors together take $O(n)$ space.
\end{proof} 

The query algorithm to construct visibility tree data structures $TVIS_{\cal P'}^B(q)$ and $TVIS_{\cal P'}^U(q)$ for the input query point $q$ among the polygonal obstacles of $\cal P'$ is described in Section~\ref{sect:queryalgo}.
The visibility polygon of $q$ is determined from the information stored at the nodes of these trees, and our dynamic algorithms update visibility trees as and when the current polygonal domain is modified with vertex insertions and vertex deletions.
Our algorithms for updating the visibility polygon of a point $q \in \cal{F(P')}$ are provided with the visibility trees $TVIS_{\cal P'}^B(q)$ and $TVIS_{\cal P'}^U(q)$.

The visibility trees were first defined in \cite{journals/siamcomp/KapoorM00}.
We modify visibility tree structures from \cite{journals/siamcomp/KapoorM00} so that they are helpful in our context. 
Here we describe the structures from \cite{journals/siamcomp/KapoorM00} together with our modifications.
For every corridor $C$ that has at least one point on the boundary of $C$ that is visible to $q$, there exists at least one node in these trees that corresponds to $C$.
Any node $t$ in either of these trees corresponds to a corridor $C^t$.
With $t$, we store a pointer to $C^t$, a visibility cone $vc^t$ (with its apex at $q$), and two red-black balanced binary search trees (refer \cite{books/algo/Cormen09}), denoted by $RBT_L^t, RBT_R^t$.
Refer to Fig.~\ref{fig:datastr}.
The $RBT_L^t$ (resp. $RBT_R^t$) at node $t$ stores every (constructed) vertex $v'$ of $VP_{\cal P'}(q)$ that belongs to the left (resp. right) side of corridor $C$ whenever $v'$ lies in $vc^t$.
With each point $p$ in both of these RBTs, we store the angle ray $qp$ makes at $q$.
If a point $p$ is stored in $TVIS_{\cal P'}^B(q)$ (resp. $TVIS_{\cal P'}^U(q)$) then the line segment $qp$ is guaranteed to intersect $B$ (resp. $U$).
For any point $p$ located on the boundary of an obstacle, and for $p$ visible to $q$, the sequence of corridors intersected by the line segment $qp$ is said to be the {\it corridor sequence} of $qp$.
We note that for any two points $p', p''$ in a corridor $C$, with both $p'$ and $p''$ visible to $q$, the corridor sequence of $qp'$ is not necessarily same as the corridor sequence of $qp''$.
Hence, in any visibility tree of $q$, there could be more than one node that corresponds to any corridor.
However, any (constructed) vertex of $VP_{\cal P'}(q)$ (or, any vertex of $\cal P'$) appears at most once in any of the RBTs stored at the nodes of these visibility trees.
For any two nodes $t', t''$ of any visibility tree, for any point $p'$ stored in either $RBT_L^{t'}$ or $RBT_R^{t'}$, and for any point $p''$ stored in either $RBT_L^{t''}$ or $RBT_R^{t''}$, the corridor sequence of $qp'$ is not equal to the corridor sequence of $qp''$. 

For every corridor $C'$, the list of visibility cones that intersect $C'$ are stored in a red-black tree, named $VC_{C'}$.
In specific, the visibility cones in $VC_{C'}$ are stored in sorted order with respect to angle left bounding ray of each cone in $VC_{C'}$ makes at $q$.
In addition, with each visibility cone $vc$ in $VC_{C'}$, we store the pointer to a node in a visibility tree that saved a visible point belonging to $bd(C') \cap vc$.
(Refer to Fig.~\ref{fig:datastr}.)
If no such visible point exists, then the pointer to a node in the visibility tree that represents the corridor nearest to $p$ along $qp$, where $p \in C' \cap vc$, is stored with $vc$.
Whenever a vertex $v$ is inserted (resp. deleted) to (resp. from) $C'$, by searching in $VC_{C'}$, we determine the visibility cone in which we need to update the visibility polygon of $q$.

\begin{figure}[h]
\centering
\minipage{0.4\textwidth}
\includegraphics[width=\linewidth]{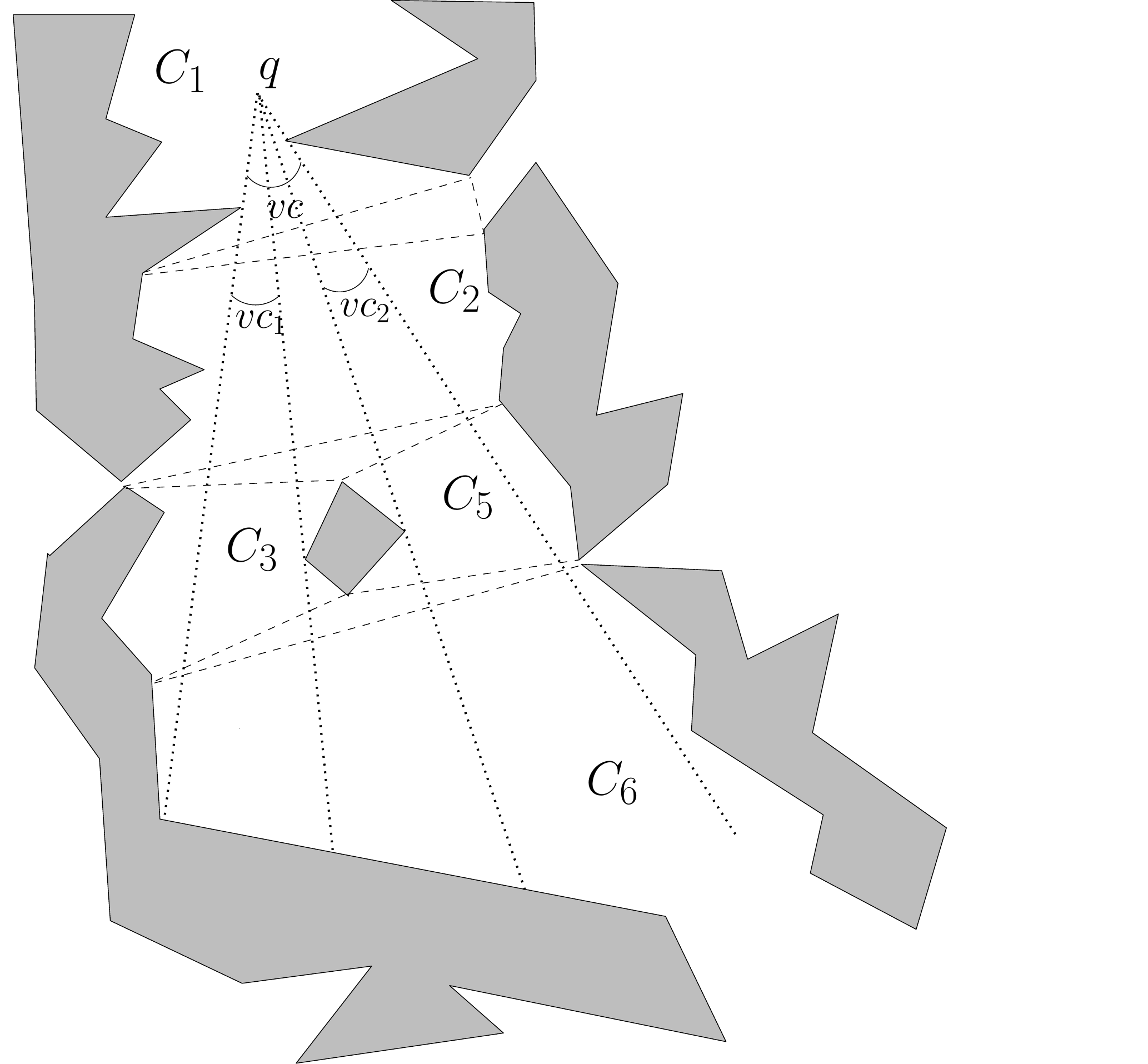}
\endminipage
\minipage{0.55\textwidth}
\includegraphics[width=1.05\linewidth]{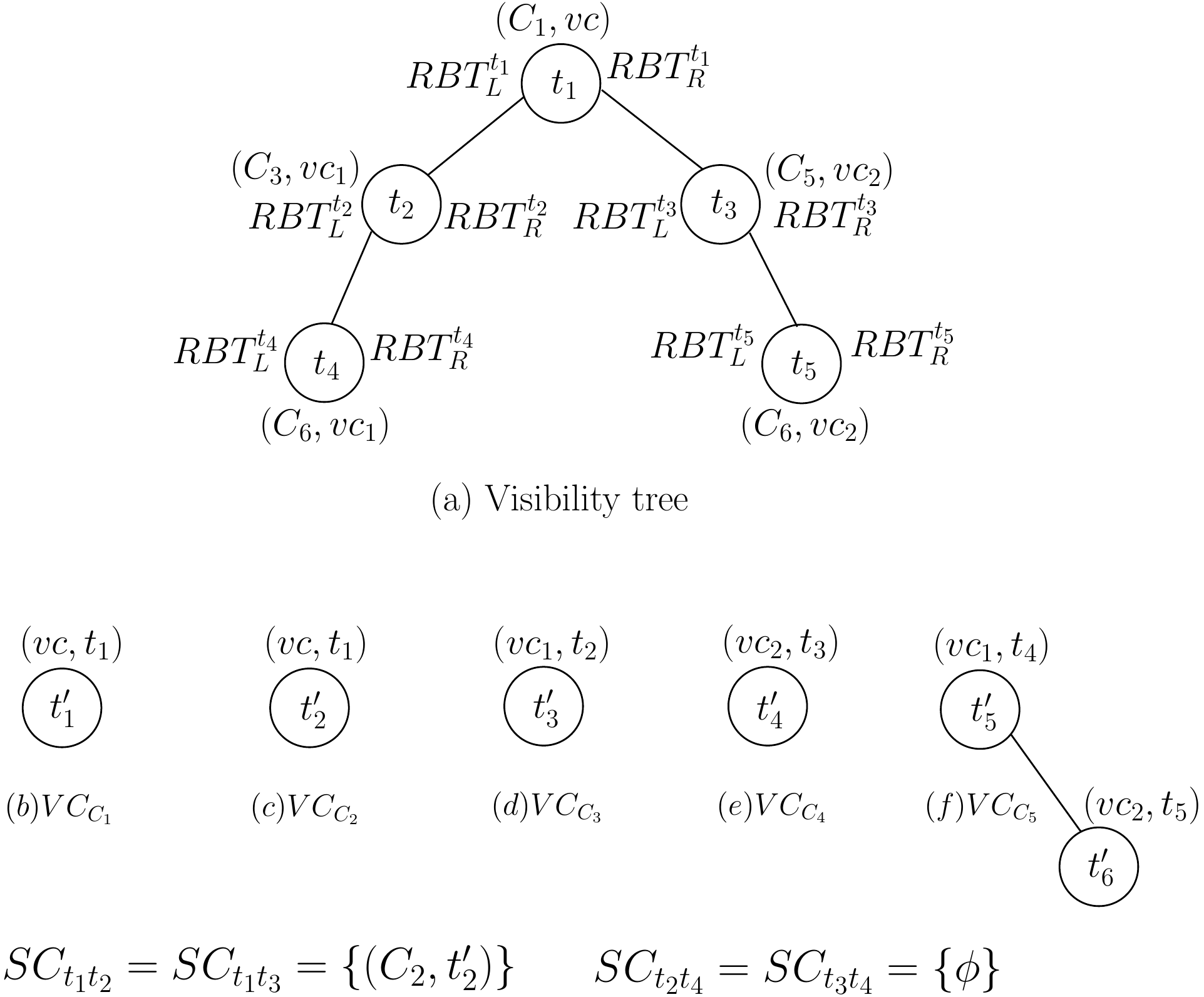}
\endminipage
\caption{Illustrating visibility tree (right top), $VC_{C_l}$ (right middle) and $SC_{t_it_j}$ (right bottom) data structures corresponding to a set of corridors (left).}
\label{fig:datastr}
\end{figure}
Let $t'$ be any node in either of the visibility trees of $q$.
And, let $t''$ be any child of $t'$.
Also, let $C', C''$ be the corridors associated to $t', t''$ respectively.
We note that it is not necessary for corridors $C'$ and $C''$ to be adjacent in the corridor subdivision of $\cal{F(P')}$.
If $C'$ and $C''$ are not adjacent, then there exists a unique sequence of corridors between $C'$ and $C''$, and this sequence of corridors is stored in a list $SC_{t't''}$.
The list $SC_{t't''}$ is associated with the edge $t't''$ of the visibility tree.
(Refer to Fig.~\ref{fig:datastr}.)
Note that if $C'$ and $C''$ are adjacent in the corridor subdivision of $\cal{F(P')}$, then the list $SC_{tt'}$ would be empty.
With every corridor $C \in SC_{t't''}$, we store a pointer to the node in the visibility tree that corresponding to visibility cone $vc^{t'}$.
In addition, the pointer stored with $vc^{t'}$ in $VC_C$ points to node $t'$.

\section{Maintaining the visibility polygon of any query point}
\label{sect:insertdelete}

Let $\cal P'$ be the polygonal domain just before inserting (resp. deleting) $v$ to (resp. from) the boundary of an obstacle.
Also, let $v_i$ and $v_{i+1}$ be the vertices between which $v$ is located.
We assume the triangle $vv_iv_{i+1}$ is located in one corridor. 
In this section, we devise algorithms for updating the data structures that save the visibility polygon of any given query point when $\cal P'$ is modified with a vertex insertion or a vertex deletion.
Mainly, we update the satellite data associated with visibility trees $TVIS_{\cal P'}^U(q)$ and $TVIS_{\cal P'}^B(q)$.
We first describe parts of the algorithm that are common to both the insertion and deletion algorithms. 
If $v$ is inserted to an obstacle of $\cal P'$, then we insert $v$ at its corresponding position into at most three simple polygons in $P_C$; note that one simple polygon in $P_C$ correspond to the side of $C$ to which $v$ is incident. 
If $v$ is deleted from an obstacle in $\cal P'$, then for every simple polygon $P \in P_C$, we delete $v$ from $P$ if $v \in P$.
Then, for each simple polygon in $P_C$ that got modified, we update the preprocessed data structures needed for determining visibility in dynamic simple polygons \cite{journals/ijcga/InkuluST20}.
Using the algorithm in \cite{journals/jal/GoodrichT97} for ray-shooting in dynamic simple polygons, we determine whether $v$ is visible to $q$ among obstacles in $\cal P'$.
This is accomplished using simple polygons in $P_C$: if $q \in C$, the ray-shooting query with ray $qv$ is performed in $P_4(C)$; otherwise, if $v$ belongs to a side $S_1$ (resp. $S_2$) of a corridor $C' (\ne C)$, we query with ray $qv$ in each simple polygon in $P_{C'}$ that corresponds to side $S_1$ (resp. $S_2$).
From the correctness of characterizations in \cite{journals/comgeo/InkuluK09}, it is immediate that we correctly determine whether $v$ is visible to $q$. 
If $v$ is found to be not visible to $q$, then $VP_{\cal P'}(q)$ does not change.
In this case, we only update the preprocessed data structures for hull trees of sides of corridor $C$, as well as the data structures for dynamic point location.
We note that all the updations of preprocessed data structures can be accomplished in $O((\lg n)^2)$ time.

Consider the case when $v$ is visible to $q$.
In this case, the insertion of $v$ (resp. deletion of $v$) may cause the deletion of (resp. insertion of) some vertices from (resp. to) the current visibility polygon $VP_{\cal P'}(q)$.
For every two vertices $v', v'' \in \{v, v_i, v_{i+1}\}$, we determine the (smaller) angle between rays $qv'$ and $qv''$.
Among these three possible cones, we find the cone $vc_m$ with the maximum cone angle.
The visible vertices belonging to $vc_m$ are the potential candidates to be deleted (resp. inserted) from (resp. to) $VP_{\cal P'}(q)$ in the insertion (resp. deletion) algorithm.
Let $B_{C_q}$ be the lower bounding edge of corridor $C_q$ containing $q$.
Without loss of generality, suppose $vc_m$ intersects $B_{C_q}$.
Noting that $v \in C$, by searching in $VC_C$, we determine the visibility cone $vc$ in which $qv$ lies. 
We let $t'$ be the node saved with $vc$ in $VC_C$, and let $C'$ be the corridor referred by $t'$.
In the following subsections, we describe the specific details of insertion and deletion algorithms. 

\subsection{\textbf{Insertion of a vertex}} 
\label{subsect:insertvert}

If corridor $C'$ is same as corridor $C$ (to remind, $v$ is in corridor $C$), then $v$ is inserted to $RBT_L^{t'}$ (resp. $RBT_R^{t'}$) of node $t'$ if $v$ is located on the left (resp. right) side of $C'$.
In the other sub-case, $C'$ is not the same as $C$.
This indicates there is no node present in $TVIS^B_{\cal P'}(q)$ that corresponds to $C$ and $vc$, i.e., before the insertion of $v$, there was no point of $bd(C) \cap vc$ is visible to $q$.
Let $t_l', t_r'$ be the left and right children of $t'$ respectively.
A new node $t''$ is inserted as a left (resp. right) child of $t'$, if $v$ is located on the left (resp. right) side of $C$, and the parent of $t_l'$ (resp. $t_r'$) is changed to $t''$.
The visibility cone $vc'$ associated with $t''$ is same as the visibility cone $vc$ associated with node $t'$.
Without loss of generality, suppose $t''$ is inserted as the left child of $t'$.
Let $C''$ be the corridor associated to node $t_l'$.
The sequence of corridors $SC_{t't_l'}$ associated to edge $t't_l'$ is splitted into two sequences: the corridor sequence between $C'$ to $C$ along visibility cone $vc$ is associated to edge $t't''$, and the corridor sequence between $C$ and $C''$ along the visibility cone $vc$ is saved with edge $t''t_l'$.
In addition, for each corridor $C'''$ in $SC_{t''t_l'}$, the pointer stored with visibility cone $vc$ in $VC_{C'''}$ is modified so that it points to node $t''$.

\begin{figure}[h]
\centering
\minipage{0.3\textwidth}
\includegraphics[width=\linewidth]{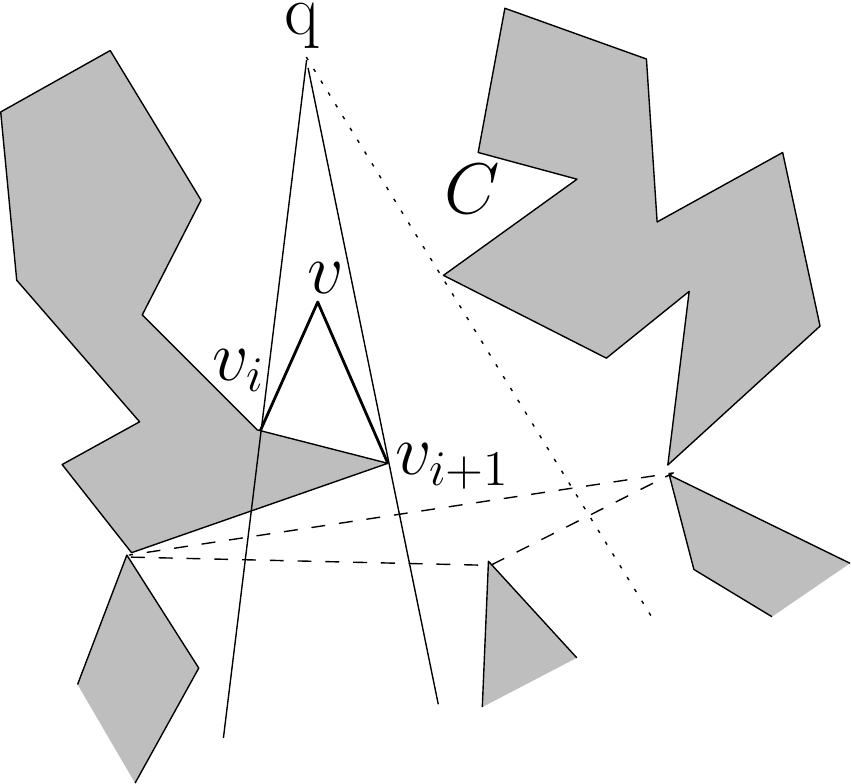}
\caption{Illustrating the case in which the rays bounding $vc_m$ are $qv_i$ and $qv_{i+1}$.}
\label{fig:ins1}
\endminipage\hfil
\minipage{0.3\textwidth}
\includegraphics[width=\linewidth]{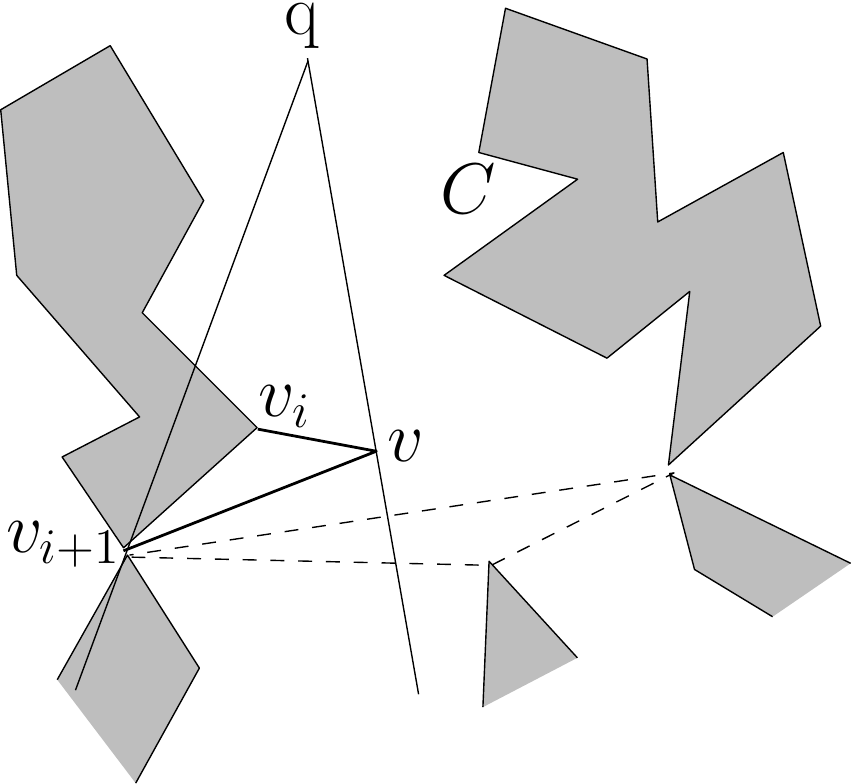}
\caption{Illustrating the case in which rays $qv_{i+1}$ and $qv$ bound $vc_m$.}
\label{fig:ins2}
\endminipage
\end{figure}

Suppose the rays bounding $vc_m$ are $qv_i$ and $qv_{i+1}$. 
(Refer to Fig.~\ref{fig:ins1}.)
Then, every vertex of $VP_{\cal P'}(q)$ continues to be visible to $q$.
Hence, there is no vertex to be deleted from $VP_{\cal P'}(q)$.
However, since $v$ is visible to $q$, we need to insert $v$ into a RBT of $TVIS^B_{\cal P'}(q)$. 
When $v$ is located on the left (resp. right) side of $C$, if $C'$ is same as $C$, then $v$ is inserted to $RBT_L^{t'}$ (resp. $RBT_R^{t'}$); otherwise, if $C'$ is not same as $C$, then vertex $v$ is inserted to $RBT_L^{t''}$ (resp. $RBT_R^{t''}$).
    
\begin{observation}
\label{obs:ins}
Let $\cal P'$ be the current polygonal domain. 
Let $VP_{\cal P'}(q)$ be the visibility polygon of a point $q \in \cal P'$.
Whenever a new vertex $v$ is inserted to an obstacle of $\cal P'$, the set of vertices of $VP_{\cal P'}(q)$ that get hidden due to the insertion of $v$ are contiguous along the boundary of $VP_{\cal P'}(q)$. 
In specific, vertices stored in any red-black tree of any visibility tree hidden due to the insertion of $v$ are contiguous at the leaves.
\end{observation}
\begin{proof}
Let $S = \{v_1, v_2, \ldots, v_j\}$ be the ordered set of vertices stored at a leaf of a red-black tree in $TVIS_{\cal P'}^B(q)$ (resp. $TVIS_{\cal P'}^U(q)$). 
The observation is immediate as the vertices in $S$ that fall in $vc_m$ is a subsequence of $S$.
\end{proof} 

Suppose the rays bounding $vc_m$ are $qv_{i+1}$ and $qv$.
(Refer to Fig.~\ref{fig:ins2}.)
In this case, due to the insertion of $v$, some vertices of $VP_{\cal P'}(q)$ may become not visible to $q$. 
To determine these vertices, we do the depth-first traversal of $TVIS_{\cal P'}^B(q)$, starting from $t'$ if $C'$ is same as $C$; otherwise, we do the depth-first traversal of $TVIS_{\cal P'}^B(q)$, starting from node $t''$.
Let $\alpha_1$ and $\alpha_2$ be the angles made by rays $qv_{i+1}$ and $qv$, respectively at $q$. 
Also, let $\alpha_1 < \alpha_2$.
For every red-black tree $T$ at every node $t$ encountered in this traversal, we search in $T$ to find the contiguous list of vertices such that each vertex in that list lies in the cone $vc_m$.
By Observation~\ref{obs:ins}, all the vertices belonging to this list are the ones that needed to be removed from $T$.
Hence, we remove each vertex $v'$ in this list from $T$, as $v'$ is no more visible to $q$.
Let $C^t$ be the corridor referred by node $t$.
During traversal, if visibility cone $vc^t$ associated with node $t$ is found to be lying completely inside $vc_m$, we delete the node corresponding to visibility cone $vc^t$ in $VC_{C^t}$.
    
The handling of the last case in which the $vc_m$ is bounded by rays $qv_{i}$ and $qv$ is analogous to the case in which $vc_m$ is bounded by rays $qv_{i+1}$ and $qv$.

\begin{lemma}
Let $\cal P''$ be the polygonal domain resultant from inserting a vertex $v$ to an obstacle of a polygonal domain $\cal P'$.
Also, let $VP_{\cal P''}(q)$ be the visibility polygon of $q$ among obstacles in $\cal P''$ determined by the algorithm.
Then, any point $p \in VP_{\cal P''}(q)$ if and only if $p$ is visible to $q$ among obstacles in $\cal P''$.
\end{lemma}
\begin{proof}
Let a vertex $v$ be inserted on a side of corridor $C \in \cal P'$.
If $v$ is visible to $q$, we search in $VC_C$ to find the visibility cone $vc$ in which $v$ is lying.
A new node is inserted in the current visibility tree of $q$ only if there is no node in that visibility tree that corresponds to corridor $C$ and visibility cone $vc$.
This ensures that inside a visibility cone, whenever a corridor has at least one visible point on its side, there exists a node corresponding to it in one of the visibility trees.
The red-black trees stored at any node of visibility trees that corresponds to $C$ and $vc$ stores all the vertices in $\cal P'$ and the constructed vertices that belong to $C \cap vc$.
Among the vertices stored in these red-black trees, vertices that are residing inside the cone $vc_m$ become not visible to $q$ after the insertion of $v$.
After inserting $v$ into the current visibility tree, we traverse the updated visibility tree and determine these vertices.
At every red-black tree $T$ encountered during the traversal, we search for the contiguous section of vertices (Observation~\ref{obs:ins}) residing in $vc_m$ and remove it from the current visibility polygon.
It ensures that any vertex belonging to $VP_{\cal P'}(q)$ is removed from the current visibility tree only if it has become invisible due to the insertion of $v$.
\end{proof}
        
\begin{lemma} 
Whenever a vertex $v$ is inserted to $\cal P'$, our algorithm updates the current visibility polygon $VP_{\cal P'}(q)$ of $q$ in $O(k(\lg|VP_{\cal P'}(q)|)+(\lg n')^2+h)$ time.
Here, $k$ is the number of combinatorial changes in $VP_{P'}(q)$ due to the insertion of $v$, $n'$ is the number of vertices of $\cal P'$, and $h$ is the number of obstacles in $\cal P'$.
\end{lemma}
\begin{proof}
Locating the corridor $C$ in which $v$ is inserted takes $O((\lg{n'})^2)$ time.
Determining whether $v$ is visible to $q$ with a ray-shooting query in each of the simple polygons in $P_C$ takes $O((\lg{n})^2)$ time.
By searching in $VC_C$, we determine the visibility cone in which $v$ is lying. 
Every node in $VC_C$ points to a node in either $TVIS_{\cal P'}^B(q)$ or $TVIS_{\cal P'}^U(q)$. 
Since the total number of nodes in either of these trees is $O(h)$, searching $VC_C$ takes $O(\lg h)$ time.
Updating pointers when a new node is inserted into a visibility tree of $q$ takes $O(h)$ time.
Traversing red-black trees at the nodes of visibility trees and removing the vertices not visible to $q$ takes overall $O(k\lg{|VP'(q)|})$ time.
Updating data structures constructed in the preprocessing phase takes $O((\lg n')^2)$ time.
\end{proof}

\subsection{\textbf{Deletion of a vertex}}
  
Suppose the rays bounding $vc_m$ are $qv_i$ and $qv_{i+1}$. 
(Refer to Fig.~\ref{fig:del1}.)
The deletion of $v$ does not change the visibility of any vertex belonging to $\cal P'$.
Hence, there is no vertex needs to be included in $VP_{\cal P'}(q)$.
However, since $v$ is not going to be visible to $q$, we need to delete $v$ from node $t'$ of $TVIS^B_{\cal P'}(q)$.
The vertex $v$ will be deleted from $RBT_L^{t'}$ (resp. $RBT_R^{t'}$) if it is located on the left (resp. right) side of $C$.
We determine whether vertices $v_i$ and $v_{i+1}$ are visible to $q$.
If any of them is not visible to $q$, then using the algorithm in \cite{journals/ijcga/InkuluST20} for dynamic simple polygons and with cone $vc_m$ as input, we determine the endpoints of constructed edges (refer to \cite{books/visalgo/skghosh2007}) residing in $vc_m$ which are incident to edge $v_iv_{i+1}$.
In addition, we insert these endpoints into $RBT_L^{t'}$ (resp. $RBT_R^{t'}$) if they are located on the left (resp. right) side of $C$.

\begin{figure}[h]
\centering
\minipage{0.3\textwidth}
\includegraphics[width=\linewidth]{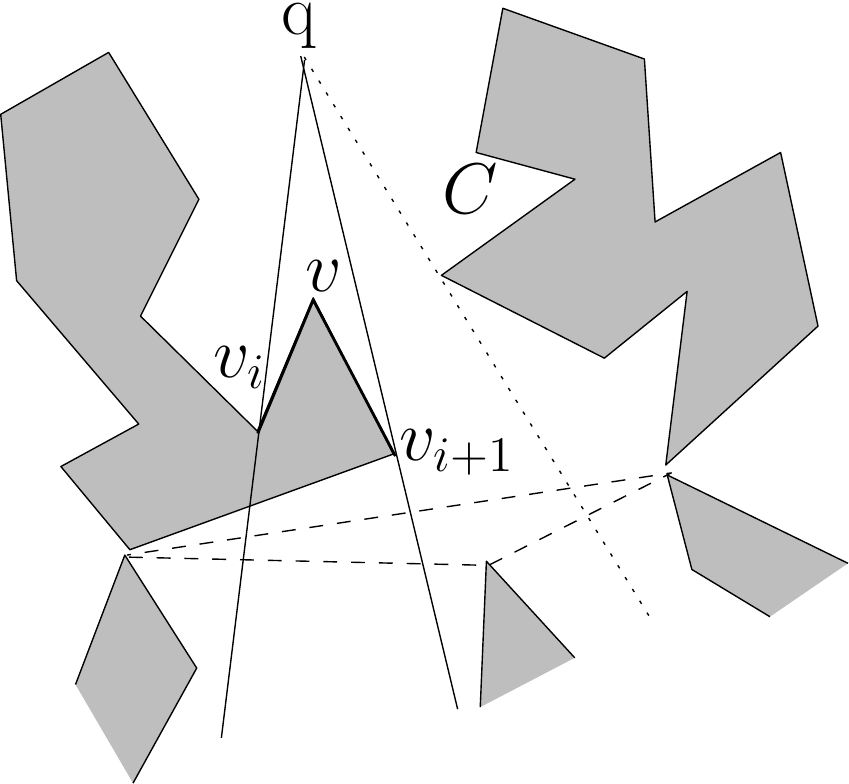}
\caption{Illustrating the case in which the rays bounding $vc_m$ are $qv_i$ and $qv_{i+1}$.}
\label{fig:del1}
\endminipage\hfil
\minipage{0.3\textwidth}
\includegraphics[width=\linewidth]{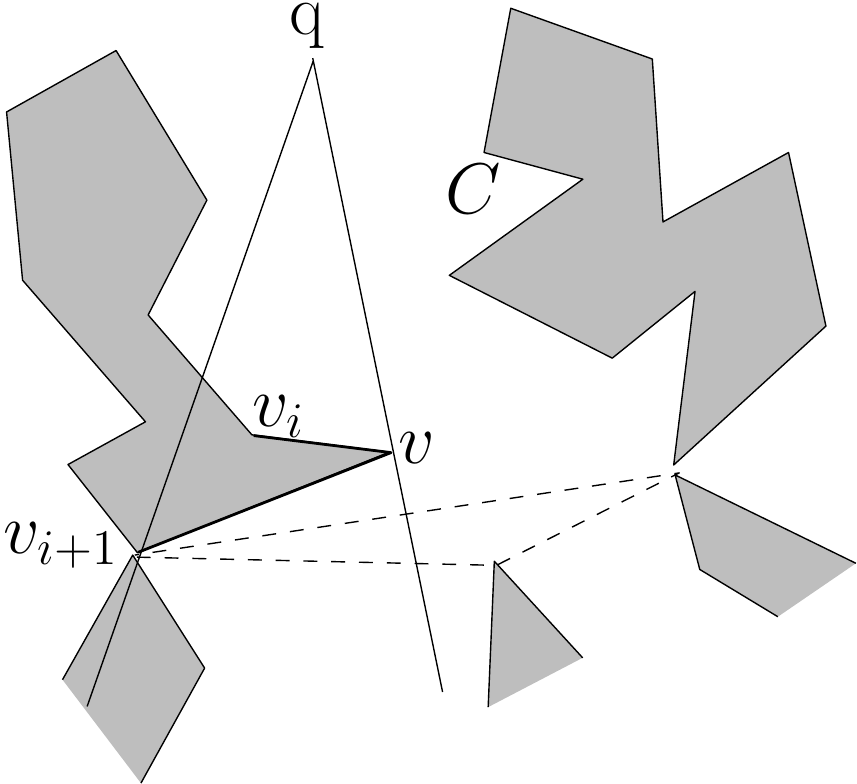}
\caption{Illustrating the case in which rays $qv_{i+1}$ and $qv$ bound $vc_m$.}
\label{fig:del2}
\endminipage
\end{figure}

Suppose the rays bounding $vc_m$ are $qv_{i+1}$ and $qv$.
(Refer to Fig.~\ref{fig:del2}).
In this case, due to the deletion of $v$, some new vertices of $\cal P'$ may become visible to $q$.
As in the above case, $v$ is deleted from $RBT_L^{t'}$ (resp. $RBT_R^{t'}$) if it is located on the left (resp. right) side of $C$.
Now, to determine vertices that become visible due to the deletion of $v$, we invoke the query algorithm described in Section~\ref{sect:queryalgo} with $vc_m \cap vc$ as the visibility cone and $q$ as the query point.
Let $vc'$ be the visibility cone $vc_m \cap vc$.
Let $vc'_l$ (resp. $vc'_r$) be the left (resp. right) bounding ray of $vc'$.
Also, let $\cal C$ be the set comprising of corridors such that the path in $G_d$ from a node of $G_d$ that corresponds to $R$ to the node of $G_d$ that corresponds to $C_q$.
Given $C$ and $vc'$, our query algorithm determines all the visible vertices on the sides of each corridor in $\cal C$ to $q$. 
(Refer to Subsection~\ref{subsect:processstack}.)
   
Let $T'$ be the tree returned by the query algorithm, and $T_{t'}$ be the subtree rooted at $t'$ in $TVIS_{\cal P'}^B(q)$. 
(Note that $t'$ is the node from which $v$ is deleted.)
For any corridor $C_i$, let $VC_{C_i}$ (resp. $VC'_{C_i}$) be the red-black tree storing the pointers to the nodes belonging to $T_{t'}$ (resp. $T'$).
Without loss of generality, suppose vertex $v$ is deleted from the left side of the corridor $C$.
Let $rr_{vc}$ be the right bounding ray of visibility cone $vc$.
For every node $t''$ in $T'$, if the right bounding ray of its visibility cone is $vc'_r$, we replace it with the ray $rr_{vc}$.
We also make similar changes in $VC'_{C''}$ corresponding to the corridor $C''$ referred by $t''$.

Further, every vertex belonging to $T'$ is added to $TVIS_{\cal P'}^{B}(q)$.
To accomplish this, we traverse the trees $T'$ and $T_{t'}$ in the breadth-first order.
For every node $t''$ that we encounter in the breadth-first traversal of $T'$, we search in $VC_{C^{t''}}$ to find a visibility cone which is lying entirely inside the visibility cone $vc^{t''}$.
If such a cone exists, it indicates that node $t''$ is present in $T_{t'}$. 
In this case, by traversing $T_{t'}$ in breadth-first order, we locate $t''$ in $T_{t'}$.
Significantly, for every $t''$, breadth-first traversal of $T_{t'}$ starts from the node where the traversal in that tree was stopped in the previous search.
The red-black trees stored at the node $t''$ in $T'$ are merged with the red-black trees stored at the node found by the breadth-first traversal in $T_{t'}$.

In the other sub-case, if node $t''$ is not present in $T_{t'}$, it indicates that before deletion, the corridor represented by this node had no visible vertex on either of its sides in visibility cone $vc$.
Due to the deletion of $v$, some portion of $bd(C_{t''})$ became visible.
Hence, $t''$ is inserted in $T_{t'}$ such that the parent of $t''$ in $T_{t'}$ is same as the parent of $t''$ in $T'$.
At the end, for every corridor $C_i$, we merge $VC'_{C_i}$ with $VC_{C_i}$.

The handling of the last case in which the $vc_m$ is bounded by rays $qv_i$ and $qv$ is analogous to the case in which $vc_m$ is bounded by $qv_{i+1}$ and $qv$.
     
\begin{lemma}
Let $\cal P''$ be the polygonal domain resultant from deleting a vertex $v$ from a polygonal domain $\cal P'$.
Also, let $VP_{\cal P''}(q)$ be the visibility polygon of $q$ among obstacles in $\cal P''$ determined by the algorithm.
Then, any point $p \in VP_{\cal P''}(q)$ if and only if $p$ is visible to $q$ in $\cal P''$.
\end{lemma}
\begin{proof}
Let a vertex $v$ be deleted from a side of corridor $C \in \cal P'$.
If $v$ is visible to $q$, we search in $VC_C$ to find the visibility cone $vc$ in which $v$ is lying.
If the rays bounding $vc_m$ are $qv_i$ and $qv_{i+1}$, we delete $v$ from a visibility tree and, using the algorithm in \cite{journals/ijcga/InkuluST20} for dynamic simple polygons, with cone $vc_m$, we determine the endpoints of constructed edges (if any) that incident to edge $v_iv_{i+1}$.
These endpoints are inserted into a visibility tree.
In the other cases, after deleting the vertex $v$ from a visibility tree, the vertices which have become visible due to the deletion of $v$ are the ones that are lying in the visibility cone $vc \cap vc_m$.
Invoking the query algorithm described in Section~\ref{sect:queryalgo} with cone $vc \cap vc_m$ ensures that a new vertex is added to the current visibility tree only if it became visible after deleting $v$.
By the correctness of query algorithm (Lemma~\ref{lem:querycorr}), it is guaranteed that every vertex of $\cal P'$ or any point on an edge of $\cal P'$, which is visible to $q$ and lying in the cone $vc \cap vc_m$, is determined by the query algorithm correctly.
While merging the new visibility tree $T'$ with the current visibility tree $T_{t'}$, any node belonging to $T'$ is inserted to $T_{t'}$ only if it is not present in $T_{t'}$.
The insertion guarantees that there are no duplicate nodes present in the updated visibility trees.
\end{proof}
     
\begin{lemma} 
Whenever a vertex $v$ is deleted from $\cal P'$, our algorithm updates the current visibility polygon $VP_{\cal P'}(q)$ of $q$ in $O(k(\lg n')^2+(\lg|VP_{\cal P'}(q)|)+h)$ time. 
Here, $k$ is the number of combinatorial changes in $VP_{P'}(q)$ due to the deletion of $v$, $n'$ is the number of vertices of $\cal P'$, and $h$ is the number of obstacles in $\cal P'$. 
\end{lemma}
\begin{proof}
Locating the corridor $C$ in which $v$ is inserted takes $O((\lg{n'})^2)$ time. 
Determining whether $v$ is visible to $q$ with a ray-shooting query in simple polygons in $P_C$ takes $O((\lg{n'})^2)$ time.
By searching in $VC_C$, we determine the visibility cone in which $v$ is lying. 
Every node in $VC_C$ points to a node in either $TVIS_{\cal P'}^B(q)$ or $TVIS_{\cal P'}^U(q)$. 
Since the total number of nodes in either of these trees is $O(h)$, searching $VC_C$ takes $O(\lg h)$ time.
Determining vertices visible due to the deletion of $v$ using the query algorithm of Section~\ref{sect:queryalgo} takes $O(k(\lg n')^2+h)$.
A breadth-first traversal of the tree returned by the query algorithm, as well as the breadth-first traversal of a visibility tree together, takes $O(h+k)$ time.  In addition, the red-black trees stored at the nodes of these trees can be updated in $O(k \lg |VP_{\cal P'}(q)|)$ time.
Searching in $VC$ data structures and updating the visibility cone at the nodes in $T'$ takes $O(k\lg h)$ time.
Updating data structures constructed in the preprocessing phase takes $O((\lg n')^2)$ time.
\end{proof}

\section{Determining the visibility polygon of a query point}
\label{sect:queryalgo}

In this section, we describe an output-sensitive algorithm to determine the visibility polygon of a query point in a dynamic polygonal domain.
We modify the algorithm for answering visibility polygon queries in \cite{journals/comgeo/InkuluK09}, so that the algorithm accommodates dynamic polygonal obstacles.

In the query phase, for any query point $q \in \cal{F(P')}$, we compute $VP_{\cal P'}(q)$.
First, we determine all the sides of corridors, each of which has at least one point visible to $q$.
Recall the visibility tree data structure from Section~\ref{sect:preprocessds}.
We store the vertices of the visibility polygon in two visibility trees, denoted by $TVIS_{\cal P'}^B(q)$ and $TVIS_{\cal P'}^U(q)$. 
For convenience, in the query algorithm, at every node in both of these visibility trees, we save an additional cone, named {\it auxiliary visibility cone}.
The auxiliary visibility cone $vc_{aux}^t$ defined at a node $t$ indicates that there is an obstacle $O \in \cal P'$ such that (i) $p \in bd(O) \cap vc_{aux}^t$, and (ii) $p$ is visible to $q$.
(The specific use of auxiliary visibility cones is described in the subsections below.)
As in \cite{journals/comgeo/InkuluK09}, for each such side $S$, the query algorithm computes all the vertices of $S$ visible to $q$.
To construct these trees, we use a stack.
This stack contains objects which are yet to be processed by the algorithm. 
Each object $obj$ in stack is represented by a tuple $[lr_{vc}, rr_{vc}, ptr_l, ptr_r, ptr_t]$.
Here, $lr_{vc}$ (resp. $rr_{vc}$) is the left (resp. right) bounding ray of the visibility cone $vc$; 
$ptr_l$ (resp. $ptr_r$) is a pointer to the first unexplored corridor in the corridor sequence of line segment $qp'$ (resp. $qp''$), where $p'$ (resp. $p''$) is the point at which $lr_{vc}$ (resp. $rr_{vc}$) strikes an obstacle $O \in \cal P'$ or the bounding box; 
and, $ptr_t$ is a pointer to the node in a visibility tree that was created at the time of initialization of $obj$.

First, using the point location data structure, we determine the corridor $C_q$ containing $q$.
Let $B_{C_q}$ (resp. $U_{C_q}$) be the lower (resp. upper) bounding edge of $C_q$.
Using hull trees, we find the points of tangency on both sides of $C_q$.
Note that there can be at most two points of tangency on each side.
For any point of tangency $p$ that is lying on the left side (resp. right side) of $C_q$, if the ray $qp$ intersects $B_{C_q}$ then point $p$ is known as $p_l^B$ (resp. $p_r^B$), and if the ray $qp$ intersects $U_{C_q}$ then point $p$ is known as $p_l^U$ (resp. $p_r^U$). 
(Refer to Fig.~\ref{fig:vis_cone}.)
    
\begin{wrapfigure}{r}{0.42\textwidth}
\centering
\includegraphics[width=0.8\linewidth]{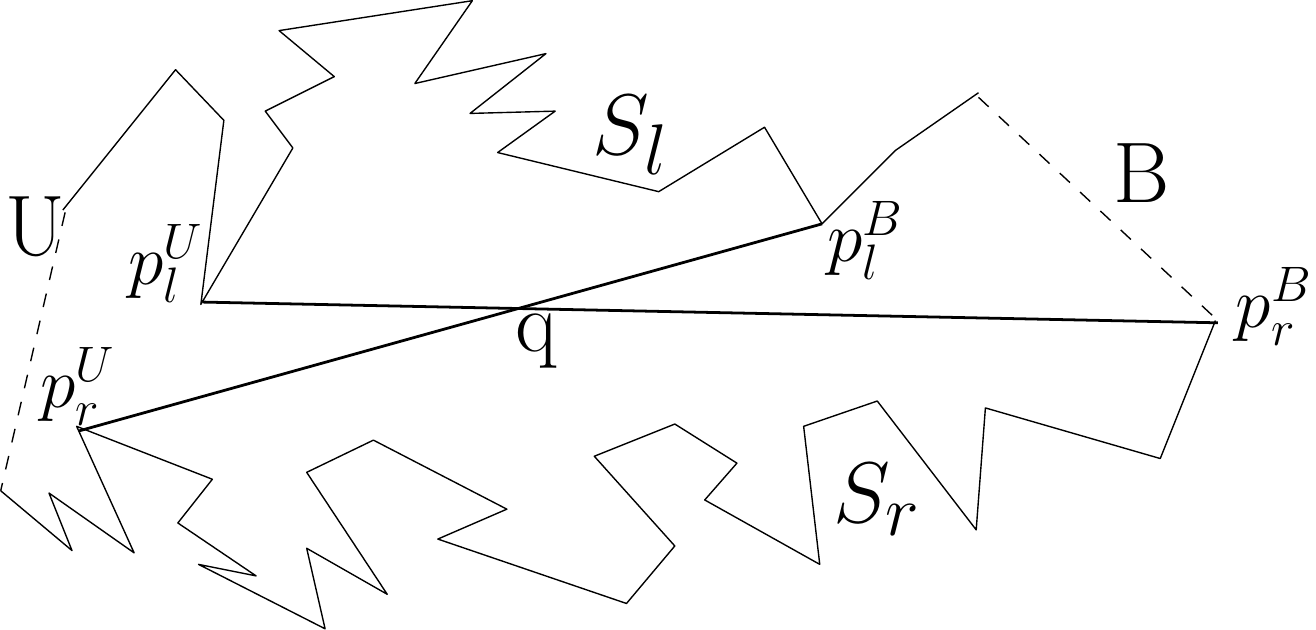}
\caption{Illustrating the initial visibility cones.}
\label{fig:vis_cone}
\end{wrapfigure}

For the current polygonal domain $\cal P'$, if the visibility cone $(p_l^B,p_r^B)$ (resp. $(p_l^U,p_r^U)$) exists, then we create a node $t^B$ (resp. $t^U$) as root node of $TVIS_{\cal P'}^B(q)$ (resp. $TVIS_{\cal P'}^U(q)$).
The node $t^B$ (resp. $t^U$) refers to corridor $C_q$, and the visibility cone $(p_l^B,p_r^B)$ (resp. $(p_l^U,p_r^U)$) is associated with $t^B$.
Also we initialize $obj^B$ (resp. $obj^U$) that corresponds to $t^B$ (resp. $t^U$). 
And, $obj^U$, followed by $obj^B$, are pushed onto the stack.

\subsection{Constructing $TVIS_{\cal P'}^B(q)$ and $TVIS_{\cal P'}^U(q)$}
\label{subsect:processstack}

The visible vertices in corridors other than corridor $C_q$ are determined by processing objects in the stack.
Let $obj = [lr_{vc}, rr_{vc}, ptr_l, ptr_r, ptr_t]$ be the object popped from the stack.
Let $CS_l$ and $CS_r$ be the corridor sequences of line segments $lr_{vc}$ and $lr_{vc}$ respectively. 
When both $ptr_l$ and $ptr_r$ refer to the same corridor $C'$, starting from $C'$ in $CS_l$ and $CS_r$, we find the last common corridor $C''$ that occurs in both $CS_l$ and $CS_r$.
For every corridor $C_i$ between $C'$ and $C''$ in $CS_l$ (or, $CS_r$), a node $t'$ associated with $ptr_t$ and $vc$ is inserted to $VC_{C_i}$, and $C_i$ saves a pointer to $t'$.
(This denotes the visibility of corridor $C_i$ is hindered by the corridor corresponding to the node pointed by $ptr_t$.)
Let $C_l$ (resp. $C_r$) be the corridor after $C''$ in $CS_l$ (resp. $CS_r$).

When the corridor $C_l$ referred by $ptr_l$ is different from the corridor $C_r$ referred by $ptr_r$, it is immediate that there is an obstacle $O$ that separates $C_l$ from $C_r$.
We find a tangent $qp_r$ to the right side of $C_l$ from $q$ and a tangent $qp_l$ to left side of $C_r$ from $q$.
We insert one node $t_l$ as the left child of $t$ which refers to $C_l$, and another node $t_r$ as the right child of $t$ which refers to $C_r$.
For every corridor $C'''$ in the sequence of corridors from $C'$ to $C''$, $C'''$ together with a pointer to $VC_{C'''}$ is associated to both the edges $tt_l$ and $tt_r$.
We also associate visibility cone $(lr_{vc}, qp_r)$ (resp. $(qp_l, rr_{vc})$) with $t_l$ (resp. $t_r$).
A node $t'$ (resp. $t''$) with $ptr_{t_l}$ (resp. $ptr_{t_r}$) and the visibility cone $(lr_{vc}, qp_r)$ (resp. $(qp_l, rr_{vc})$) is inserted to $VC_{C_l}$ (resp. $VC_{C_r}$).
In addition, an auxiliary visibility cone $(qp_r, qp_l)$ is stored at node $t$. 
We initialize $obj_l$ (resp. $obj_r$) that corresponds to $t_l$ (resp. $t_r$). 
And, $obj_r$ is pushed onto the stack, followed by $obj_l$.
(Refer to Fig.~\ref{fig:case_1}.)
        
If no corridor exists after $C''$ in $CS_l$, then we determine the point $p$ at which ray $lr_{vc}$ strikes an obstacle.
Let $C_p$ be the corridor in which $p$ is located.
We find the point of tangency $p_l$ from $q$ to the left side of $C_p$, using the hull tree corresponding to that side. 
(Refer to Fig.~\ref{fig:case_2.1}.)
One new node is inserted as left child $t_l$ of $t$ that correspond to $C_p$ and $vc$. 
A node $t'$ with $ptr_{t_l}$ and $vc$ is inserted to $VC_{C_p}$.
We initialize an object $obj$ that corresponds to $t_l$ and push that object onto the stack.
For every corridor $C'''$ in the sequence of corridors from $C'$ to $C''$, $C'''$ together with a pointer to $VC_{C'''}$ is associated to edge $tt_l$.
When tangent to the left side of $C_p$ does not exist, no object is pushed onto the stack. 
(Refer to Fig.~\ref{fig:case_2.2}.)
The algorithm for handling when there is no corridor after $C''$ in $CS_r$ is analogous.

\begin{figure}[h]
\centering
\minipage{0.27\textwidth}
\includegraphics[width=\linewidth]{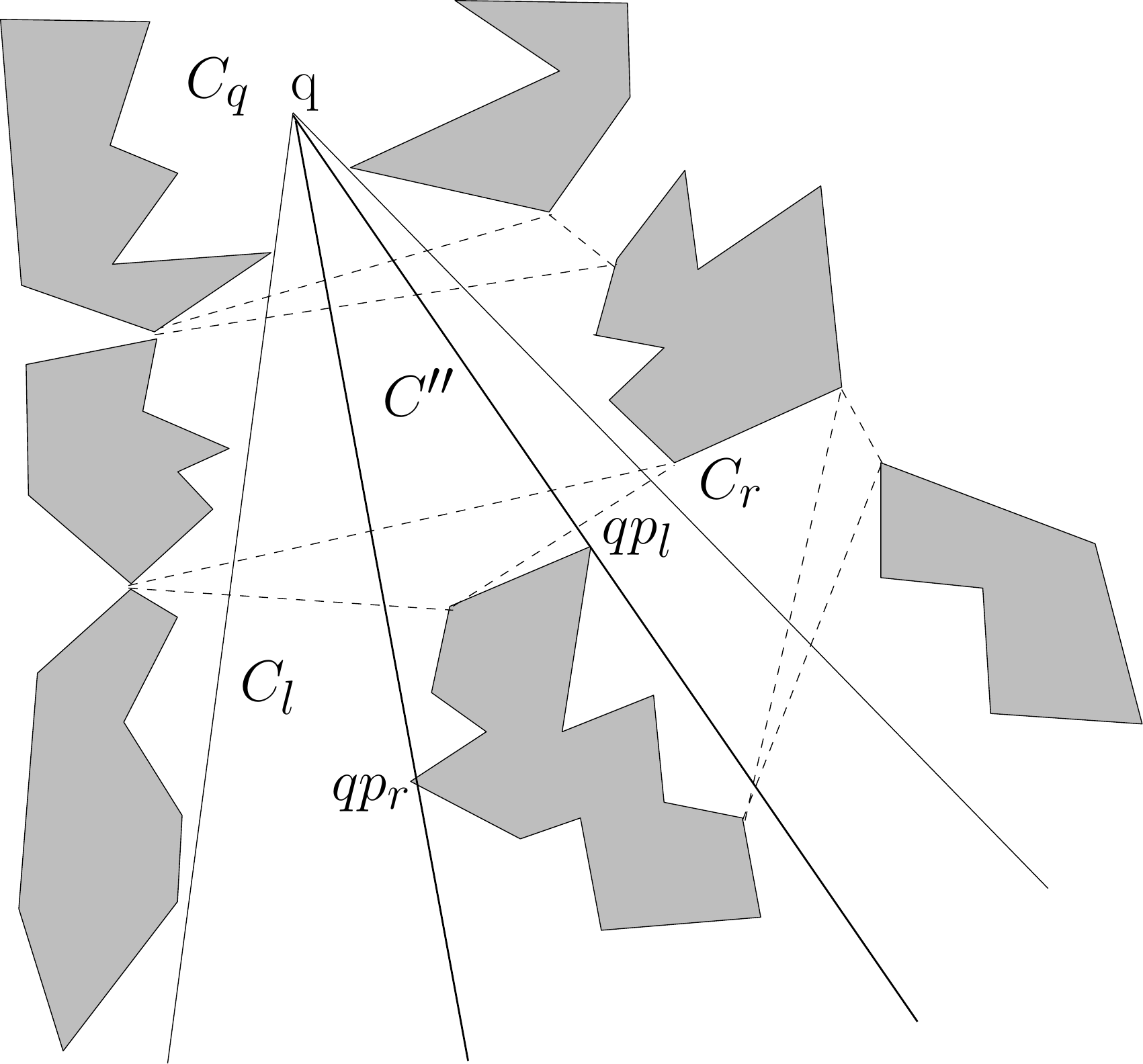}
\caption{Illustrating the case of $C_l$ being not equal to $C_r$.}
\label{fig:case_1}
\endminipage\hfill
\minipage{0.27\textwidth}
\includegraphics[width=\linewidth]{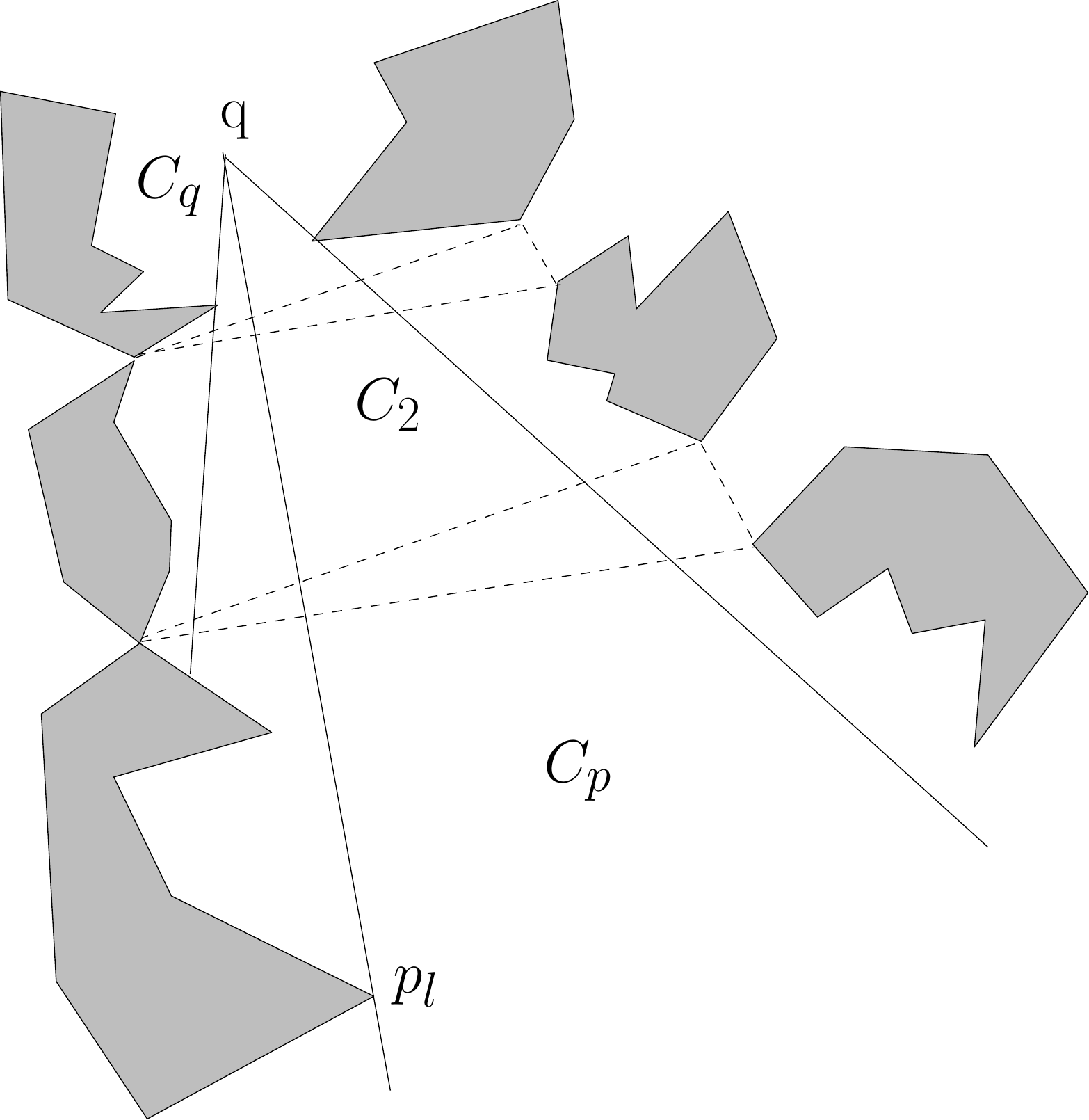}
\caption{Illustrating the ray $lr_{vc}$ striking corridor $C_p$, and $qp_l$ being a tangent to left side of $C_p$.}
\label{fig:case_2.1}
\endminipage\hfill
\minipage{0.27\textwidth}
\includegraphics[width=\linewidth]{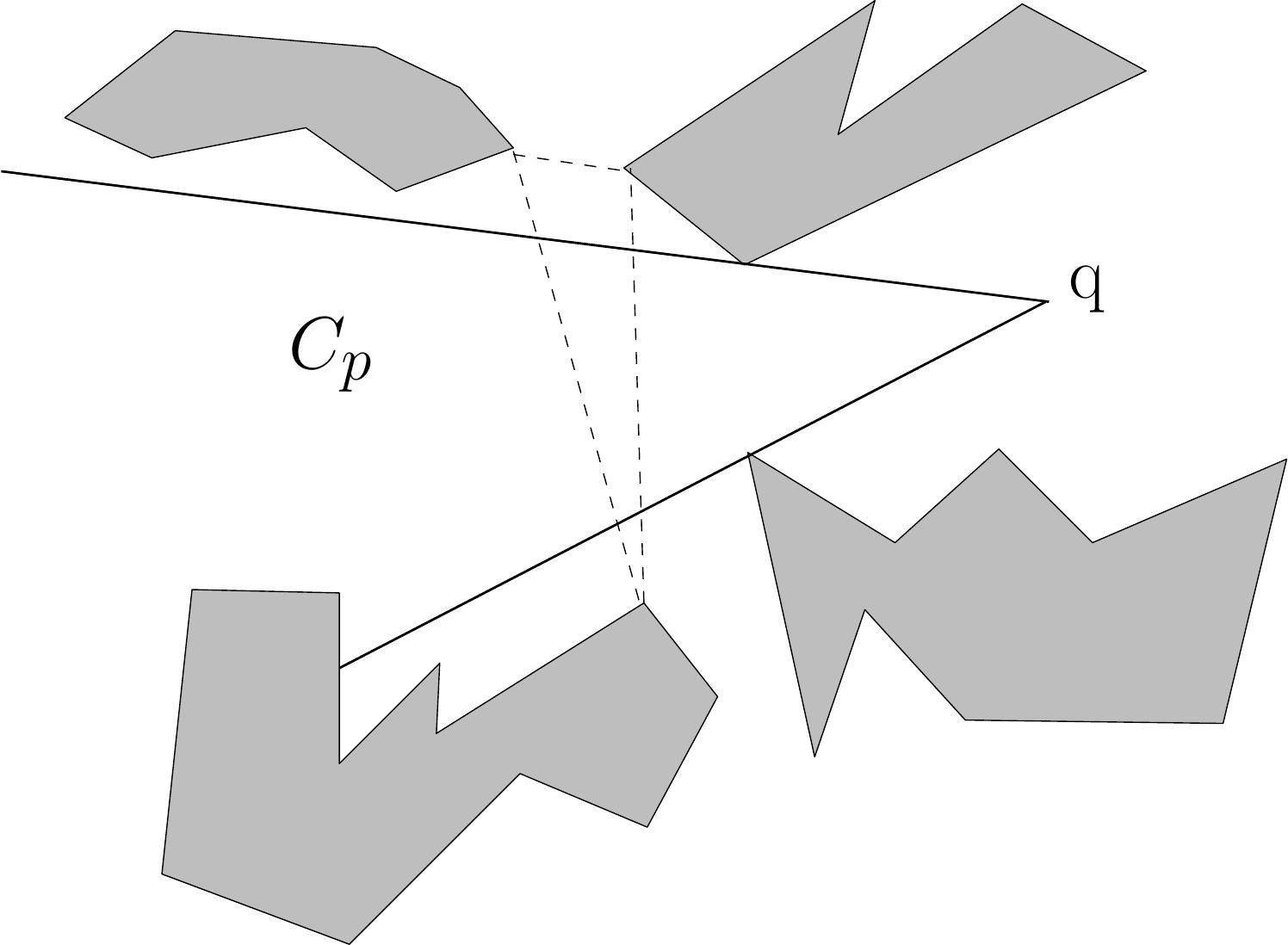}
\caption{Illustrating the ray $lr_{vc}$ striking corridor $C_p$, and tangent to left side of $C_p$ does not exist.}
\label{fig:case_2.2}
\endminipage
\end{figure}

\subsection{Computing visible vertices at the nodes of $TVIS_{\cal P'}^B(q)$ and $TVIS_{\cal P'}^U(q)$}
    
As in \cite{journals/comgeo/InkuluK09}, we traverse each of the visibility trees in depth-first order.
At every node $t$, for each side $S$ of $C_t$, we determine vertices of $VP_{\cal P'}(q)$ that belong to $S \cap vc^t$, by applying the algorithm in \cite{journals/ijcga/InkuluST20} to each simple polygon in $P_{C_t}$ that corresponds to $S$, with $q$ and $vc^t$ as the additional two parameters.
In addition, we store visible vertices on the left side (resp. right side) of $C_t$ in $RBT_{L}^{t}$ (resp. $RBT_{R}^{t}$).
After recursively traversing both the left and right subtrees of $t$, we determine all the vertices that are visible in $vc_{aux}^t$. 
Let $t_l$ and $t_r$ be the left child and right child of $t$, respectively. 
Let $C_t$ be the corridor referred by node $t$.
Let $O$ be the obstacle that lies in $vc_{aux}^t$.
Also, let $C_l$ and $C_r$ respectively be the corridors corresponding to $t_l$ and $t_r$.
In other words, $bd(O)$ is the union of right side $S_r$ of $C_l$ and the left side $S_l$ of $C_r$. 
We determine all the vertices of $S_r$ (resp. $S_l$) that are visible to $q$ and located in cone $vc_{aux}^t$, by applying the algorithm in \cite{journals/ijcga/InkuluST20} to each simple polygon in $P_{C_l}$ that corresponds to $S_r$ (resp. $S_l$), with $q$ and $vc_{aux}^t$ as the additional two parameters.
In addition, we store these visible vertices in $RBT_{R}^{t_l}$ (resp. $RBT_{L}^{t_r}$).

\subsection{Constructing $VP_{\cal P'}(q)$ using $TVIS_{\cal P'}^B(q)$ and $TVIS_{\cal P'}^U(q)$}

To construct $VP_{\cal P'}(q)$, we traverse both the visibility trees in depth-first order. First we traverse $TVIS_{\cal P'}^B(q)$ followed by $TVIS_{\cal P'}^U(q)$.
At every node $t$ encountered during the traversal, we traverse the leaf nodes of $RBT_L^t$ in left to right order and output the respective points stored at them.
Then, we recursively traverse the left subtree of $t$, followed by the right subtree of $t$.
After that, we traverse the leaf nodes of $RBT_R^t$ in the right to left order and output the respective points stored at them.
Let $S_{dft}$ be the ordered set of points obtained from the depth-first traversal of $TVIS_{\cal P'}^B(q)$ followed by $TVIS_{\cal P'}^U(q)$ as described above.
The sequence of edges obtained by joining the consecutive points in the ordered set $S_{dft}$, followed by the edge joining the first and last point in $S_{dft}$ is the boundary of visibility polygon of query point $q$.

\begin{lemma} \label{l2}
Let $\cal P'$ be a polygonal domain.
Also, let $VP_{\cal P'}(q)$ be the visibility polygon of $q$ among obstacles in $\cal P'$ determined by the algorithm.
Then, any point $p \in VP_{\cal P'}(q)$ if and only if $p$ is visible to $q$ in $\cal P'$.
\end{lemma}
\begin{proof} 
Consider a point $p$ that is visible to $q$. 
The point $p$ may be a vertex of $\cal P'$, or it may be a visible point on an edge of $\cal P'$ that will appear as an endpoint of a constructed edge of $VP_{\cal P'}(q)$.
As the algorithm starts, we determine the corridor $C_q$ containing $q$.
Let $B_{C_q}$ (resp. $U_{C_q}$) be the lower (resp. upper) bounding edge of $C_q$.
To determine the initial visibility cone, we find tangents to both the sides of $C_q$ from $q$.
Any point which is lying outside the visibility cones defined using these tangents is guaranteed to be not visible to $q$.
After this, we create the root node of $TVIS_{\cal P'}^B(q)$ (resp. $TVIS_{\cal P'}^U(q)$) if there exist a visibility cone that intersects $B_{C_q}$ (resp. $U_{C_q}$).
This root node refers to the corridor $C_q$.
It ensures that any point that is lying on either of the sides of $C_q$ and visible to $q$ is determined by the algorithm. 

Consider the other case when $p$ lies in a corridor $C'$ other than $C_q$.
Every such corridor $C'$ is determined by processing the objects popped from the stack.
Let $vc$ be the visibility cone in the current object popped from the stack.  
Whenever there is an obstacle $bd(O)$ that lies in $vc$, the tangents found on the boundaries of $bd(O)$ ensure part of that scene that is not visible due to $bd(O)$ is not considered further. 
The auxiliary cone helps in computing the points on $bd(O)$ that are visible to $q$. 
In the other case, when one of the bounding rays $r$ of $vc$ strike the side $S$ of a corridor $C_p$, we find the tangent from $q$ on the side $S$. 
Using this, we determine a section of $bd(O)$ such that no point on that section is visible to $q$.

Further, in every visibility cone for every corridor that is having at least one visible point on either of its sides, we insert a node corresponding to it in a visibility tree.
Hence, by traversing both the visibility trees, it is ensured that every vertex of $\cal P'$ or any point on an edge of $\cal P'$ that is visible to $q$, is guaranteed to be determined.
Whenever we pop an object from the stack while going through the corridor sequences of the bounding rays of the current visibility cone, for every corridor $C$ in the corridor sequence, we update $VC_{C}$ data structure.
This helps in correctly maintaining the visible cones that intersect $C$ in sorted order of their intersection with $C$.
\end{proof}

\begin{lemma}
\label{lem:querycorr}
Let $\cal P'$ be the current polygonal domain.
Given any query point $q$ in $\cal{F(P')}$, the visibility polygon of $q$ among obstacles in $\cal P'$ is computed in $O(|VP_{\cal P'}(q)|(\lg{n'})^2+h)$ time.
\end{lemma}
\begin{proof}
Using the point location query algorithm from \cite{journals/jal/GoodrichT97}, locating the corridor $C_q$ in which $q$ is located takes $O((\lg{n'})^2)$ time.
A new object is pushed onto the stack whenever an obstacle is encountered or one of the rays of visibility cone strikes the boundary of some corridor.
Since the total number of obstacles as well as corridors is $O(h)$, the overall time to push and pop objects in the stack is $O(h)$.
A node corresponding to a corridor is inserted to either of the visibility trees only if it has at least one point visible to $q$. Hence, the total number of nodes in $TVIS_{\cal P'}^B(q)$ and $TVIS_{\cal P'}^U(q)$ is $O(min(h, |VP_{\cal P'}(q)|))$. Therefore, the depth-first traversal of these trees takes $O(min(h, |VP_{\cal P'}(q)|))$ time.
Every stack element popped from stack leads to insertion of at most two nodes into either $TVIS_{\cal P'}^B(q)$ or $TVIS_{\cal P'}^U(q)$, and finding tangents corresponding to corridors at these nodes using hull trees takes $O((\lg{n'})^2)$ time. Hence, total time taken to process all the stack objects is $O(min(h, |VP_{\cal P'}(q)|)(\lg{n'})^2+h)$.
To compute the visibility polygon, algorithm from \cite{journals/ijcga/InkuluST20} is invoked at every node of $TVIS_{\cal P'}^B(q)$ as well as $TVIS_{\cal P'}^U(q)$. 
Hence, this step takes $O(|VP_{\cal P'}(q)|(\lg{n'})^2)$ time.
\end{proof}

\begin{theorem}
\label{thm:vpdynobst}
Given a polygonal domain $\cal P$ defined with $h$ obstacles and $n$ vertices, we preprocess $\cal P$ in $O(n(\lg{n})^2 + h(\lg{h})^{1+\epsilon})$ time to construct data structures of size $O(n)$ so that  
(i) whenever a vertex $v$ is inserted to the current polygonal domain $\cal P'$, the algorithm updates data structures that store visibility polygon $VP_{P'}(q)$ of a query point $q$ in $O(k(\lg|VP_{P'}(q)|)+(\lg n')^2+h)$ time,
(ii) whenever a vertex $v$ is deleted from the current polygonal domain $\cal P'$, the algorithm updates data structures that store visibility polygon $VP_{\cal P'}(q)$ of a query point $q$ in $O(k(\lg n')^2+(\lg|VP_{\cal P'}(q)|)+h)$ time,
and
(iii) whenever a query point $q$ is given, the algorithm outputs the visibility polygon in the current polygonal domain in $O(|VP_{\cal P'}(q)|(\lg{n'})^2+h)$ time.
Here, $\epsilon$ is a small positive constant resulting from the triangulation of the free space $\cal{F(P)}$ using the algorithm in \cite{journals/ijcga/Bar-YehudaC94}, $k$ is the number of combinatorial changes in $VP_{P'}(q)$ due to the insertion or deletion of $v$, and $n'$ is the number of vertices of $\cal P'$.
\end{theorem}

\section{Maintaining the visibility graph}
\label{sect:visgrmaint}

In this section, we describe an algorithm to maintain the visibility graph among dynamic polygonal obstacles in the plane.
We first detail the preprocessing algorithm with the input polygonal domain $\cal P$ with $n$ vertices and $h$ polygonal obstacles.
Our algorithm relies on the algorithm for maintaining the visibility polygon among dynamic polygonal obstacles in Section~\ref{sect:preprocessds}. 
Hence, we compute all the data structures as required in Lemma~\ref{lem:vispolypreproc}.
In addition, we need the following data structures.
Using \cite{journals/siamcomp/KapoorM00}, we construct the visibility graph of $\cal P$ in $O(|E|+h\lg n+h(\lg h)^{1+\epsilon})$ time.
Here, $E$ is the number of edges in the visibility graph of $\cal P$.
For every vertex $v \in \cal P$, we construct a red-black tree (detailed in \cite{books/algo/Cormen09}), denoted by $RBT_v$, that contains all the visible edges that are incident to $v$.
Every leaf of $RBT_v$ represents a unique visible edge that is incident to $v$.
For every visible edge that is determined to be incident to $v$, we insert it to $RBT_v$ in $O(\lg{|E|})$ time.
With every visible edge, $RBT_v$ stores the angle it makes with the positive $x$-axis. 
The left to right order of visible edges stored at the leaves of $RBT_v$ is the sorted order of visible edges incident to $v$ with respect to angle each makes with the positive $x$-axis.
Whenever a vertex is inserted to or deleted from any of the obstacles of the current polygonal domain $\cal P'$, as part of updating the visibility graph of $\cal P'$, we update these red-black tree data structures as well as the data structures for visibility polygon maintenance (refer to Theorem~\ref{thm:vpdynobst}) and return the updated visibility graph.

\begin{lemma} 
Given a polygonal domain $\cal P$ defined with $h$ holes and $n$ vertices, as part of preprocessing $\cal P$, in $O(n(\lg{n})^2+h(\lg h)^{1+\epsilon}+|E|\lg |E|)$ time, our algorithm constructs data structures of size $O(n+|E|)$.
Here, $|E|$ is the number of edges in the visibility graph of $\cal P$.
\end{lemma}

Let $\cal P'$ be the polygonal domain before inserting $v$ to (resp. deleting $v$ from) the boundary of an obstacle.
Also, let $v_i$ and $v_{i+1}$ be the vertices between which $v$ is located.
We first describe parts that are common to both the insertion and deletion algorithms. 
Using point location data structure, we determine the corridor $C$ in which $v$ is located. 
If $v$ is inserted to an obstacle of $\cal P'$, then we insert $v$ at its corresponding position into at most three simple polygons in $P'_C$ wherein each of these simple polygons has both $v_i$ and $v_{i+1}$.
If $v$ is deleted from an obstacle of $\cal P'$, then for every simple polygon $P \in P'_C$, we delete $v$ from $P$ if $v \in P$.
In addition, for each simple polygon in $P'_C$ that got modified, we update the preprocessed data structures needed to determine the visibility in dynamic simple polygons using the algorithm in \cite{journals/ijcga/InkuluST20}.

Let $v'$ be any vertex distinct from $v, v_i$, and $v_{i+1}$.
Using the algorithm in \cite{journals/jal/GoodrichT97} for ray-shooting in dynamic simple polygons, we determine whether $v$ is visible to $v'$ among obstacles in $\cal P'$.
This is accomplished using simple polygons in $P'_C$: if $v' \in C$, the ray-shooting query with ray $v'v$ is performed in $P'_4(C)$; otherwise, if $v$ belongs to a side $S_1$ (resp. $S_2$) of a corridor $C' (\ne C)$, we query with ray $v'v$ in each simple polygon in $P'_{C'}$ that corresponds to side $S_1$ (resp. $S_2$).
From the correctness of characterizations in \cite{journals/comgeo/InkuluK09}, it is immediate that we correctly determine whether $v$ is visible to $v'$. 
If $v$ is found to be not visible to $v'$, then $VP_{\cal P'}(v')$ does not change.
In this case, we only update the preprocessed data structures for hull trees of sides of corridor $C$ and the data structures for dynamic point location.
We note that all the updations of the preprocessed data structures can be accomplished in $O((\lg n)^2)$ worst-case time.
Consider the case when $v$ is visible to $v'$.
In this case, the insertion of $v$ (resp. deletion of $v$) may cause the deletion of (resp. insertion of) edges from (resp. to) $VG_{\cal P'}$ that are incident to $v'$ and intersect the triangle $vv_iv_{i+1}$.  
In addition, in the case of insertion of $v$, we introduce edge $v'v$ into $VG_{\cal P'}$.

\begin{wrapfigure}{r}{8cm}
\centering
\includegraphics[width=0.6\linewidth]{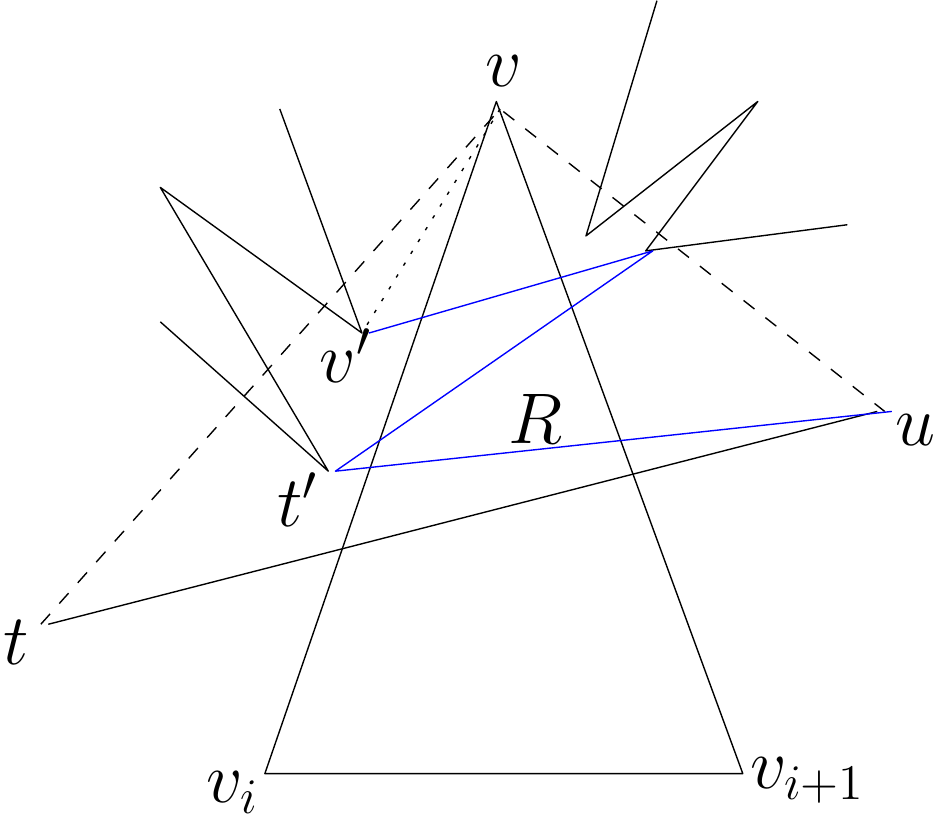}
\caption{Illustrating the notation in Lemma~\ref{lem:dynvisgr} statement. 
The path $R$ is shown in blue color.
The dotted line between $v$ and $v'$ illustrates an endpoint $v'$ of $R$ is visible to $v$.
(This illustration is from \cite{conf/caldam/Inkulu19})}
\label{fig:spp1}
\end{wrapfigure}

The following lemma statement is simplified from the one in Choudhury and Inkulu~\cite{conf/caldam/Inkulu19}, and for completeness and to show the proof extends to polygonal domains, we include a proof.
This lemma helps in efficiently determining the set of edges that intersect the triangle $v_ivv_{i+1}$.

\begin{lemma}%[Lemma 3, \cite{conf/caldam/Inkulu19}]
\label{lem:dynvisgr}
Let $v_i, v$, and $v_{i+1}$ be consecutive vertices along an obstacle of $\cal P''$.
For every visible edge $tu$ of $VG_{\cal P'}$ that intersects the triangle $vv_iv_{i+1}$, either 
(a) at least one of $t$ or $u$ is visible to $v$, or 
(b) there exists a path $R$ in $VG_{\cal P'}$ from a vertex $v'$ to $t$ (resp. $u$) such that each edge in $R$ intersects the triangle $v_ivv_{i+1}$ and $v'$ is visible to $v$.
\end{lemma}
\begin{proof}
If either $t$ or $u$ is visible to $v$, there is nothing to prove.
Otherwise, among all the edges of $VG_{\cal P'}$ that intersect both the triangles $vv_iv_{i+1}$ and $tuv$, let $t'u'$ be a visible edge in $\cal P'$ that makes the least angle with the line segment $tu$ at either $t'=t$ or $u'=u$.
Without loss of generality, we suppose $u'=u$.
(Refer to Fig.~\ref{fig:spp1}.)
Note that both $t$ and $u$ are reachable from $t'$ in $VG_{\cal P'}$.
Besides, both the edges $t'u$ and $ut$ intersect the triangle $vv_iv_{i+1}$.
Further, $t'$ is closer to $v$ as compared to $t$, with respect to Euclidean distance.
Inductively, $t'$ is reachable from $v$ in $VG_{\cal P'}$ as described in the lemma statement.
Since there are finite vertices, there exists some edge $t''u''$ such that $t''u''$ intersects $vv_iv_{i+1}$ and either $t''$ or $u''$ is visible to $v$.
\end{proof}

In the following subsections, we describe more details of insertion and deletion algorithms. 

\subsection{Insertion of a vertex} 
\label{visgraph:insert}

Let $v_i$ and $v_{i+1}$ be the vertices between which the vertex $v$ is inserted.
Also, let $S$ be the set of visible edges that intersect the triangle $vv_iv_{i+1}$.
We observe that any visible edge $e\in VG_{\cal P'}$ appears in the updated visibility graph $VG_{\cal P''}$ if and only if $e$ does not belong to set $S$.
First, we determine vertices that are visible to $v$. 
To accomplish this, we invoke the visibility polygon query algorithm among dynamic obstacles detailed in Section~\ref{sect:insertdelete}, with $v$ as the query point.
A new red-black tree, denoted by $RBT_v$ is initiated and, for every vertex $v'\in \cal P'$ determined by the visibility query algorithm, we insert the visibility edge $vv'$ into $RBT_v$ as well into $RBT_{v'}$.
This completes the insertion of new visible edges to the visibility graph $VG_{\cal P'}$.
Next, we determine the set $S$ of visible edges of $VG_{\cal P'}$ whose endpoints have become not visible to each other after the insertion of $v$.
To efficiently determine edges in $S$, we use Lemma~\ref{lem:dynvisgr} in iteratively finding all the vertices to which visible edges to be removed from $VG_{\cal P'}$ are incident.
Let $S_1$ be the set comprising of all the vertices of $\cal P'$ visible to $v$.
The points in $S_1$ are determined from the visibility polygon of $v$.
For each vertex $v' \in S_1$, we determine all the vertices $v''$ such that $v'v''$ is a visible edge in $VG_{\cal P'}$ and it intersects the triangle $vv_iv_{i+1}$.
Without loss of generality, suppose the ray $v'v_i$ makes larger angle with $v'v$ as compared to $v'v_{i+1}$.
Let $vc_m$ be the cone with rays $v'v$ and $v'v_i$ bounding it.
The visible edges lying in cone $vc_m$ are the potential candidates to be deleted from $VG_{\cal P'}$.
We search in $RBT_{v'}$ to find every vertex $v''$ such that the edge $v'v''$ belongs to cone $vc_m$, and delete the edge $v'v''$ from both $RBT_{v'}$ and $RBT_{v''}$.
In specific, for every $j \ge 1$, let $S_{j+1}$ be equal to $\bigcup_{v' \in S_j} S_{v'}$.
Following the same procedure, for every $j > 1$, for every vertex $v' \in S_j$ and $v' \notin S_{j'}$ with $j' < j$, we find the set $S_{j+1}$ comprising of all the visible edges in $\cal P'$ that are incident to $v'$ and intersect triangle $vv_iv_{i+1}$.
If at any point, $S_j \subseteq S_{j'}$ for some $j' < j$, we terminate the algorithm. 
(Refer to Fig.~\ref{fig:insert}.)

\begin{figure}[h]
\centering

\minipage{0.29\textwidth}
\includegraphics[width=\linewidth]{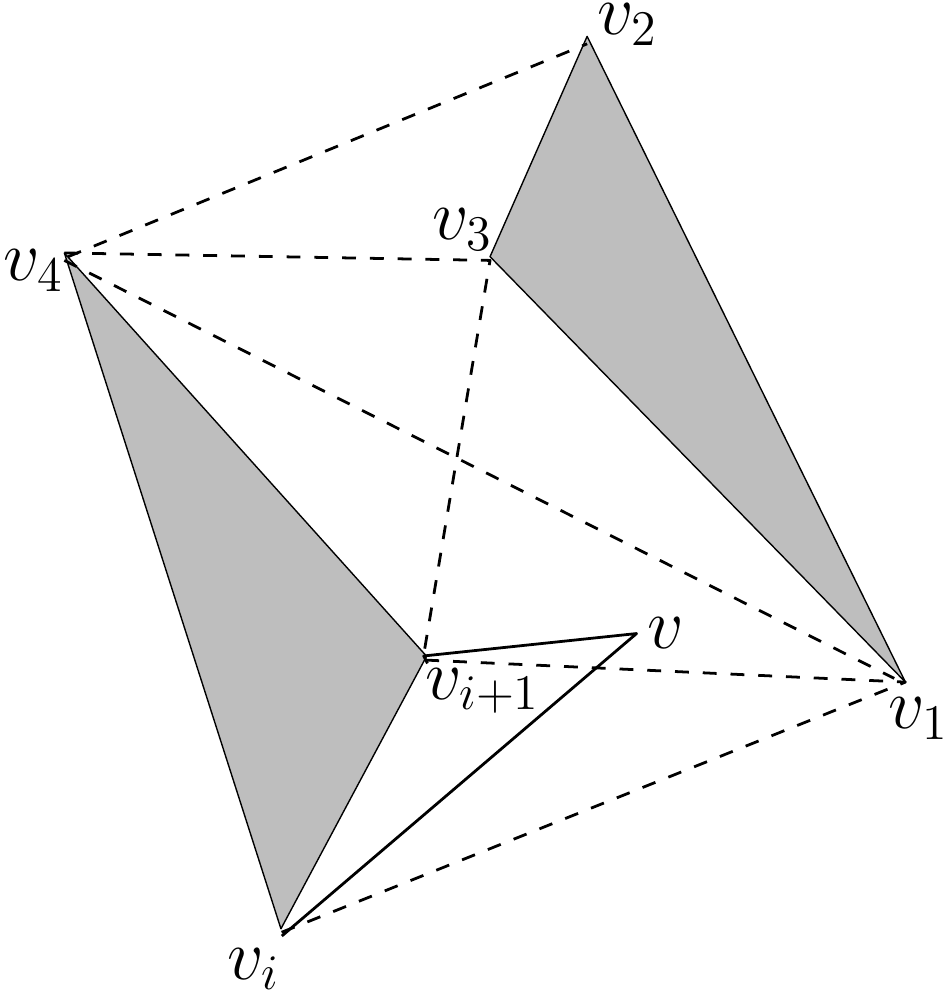}
\endminipage\hfil
\minipage{0.33\textwidth}
\includegraphics[width=\linewidth]{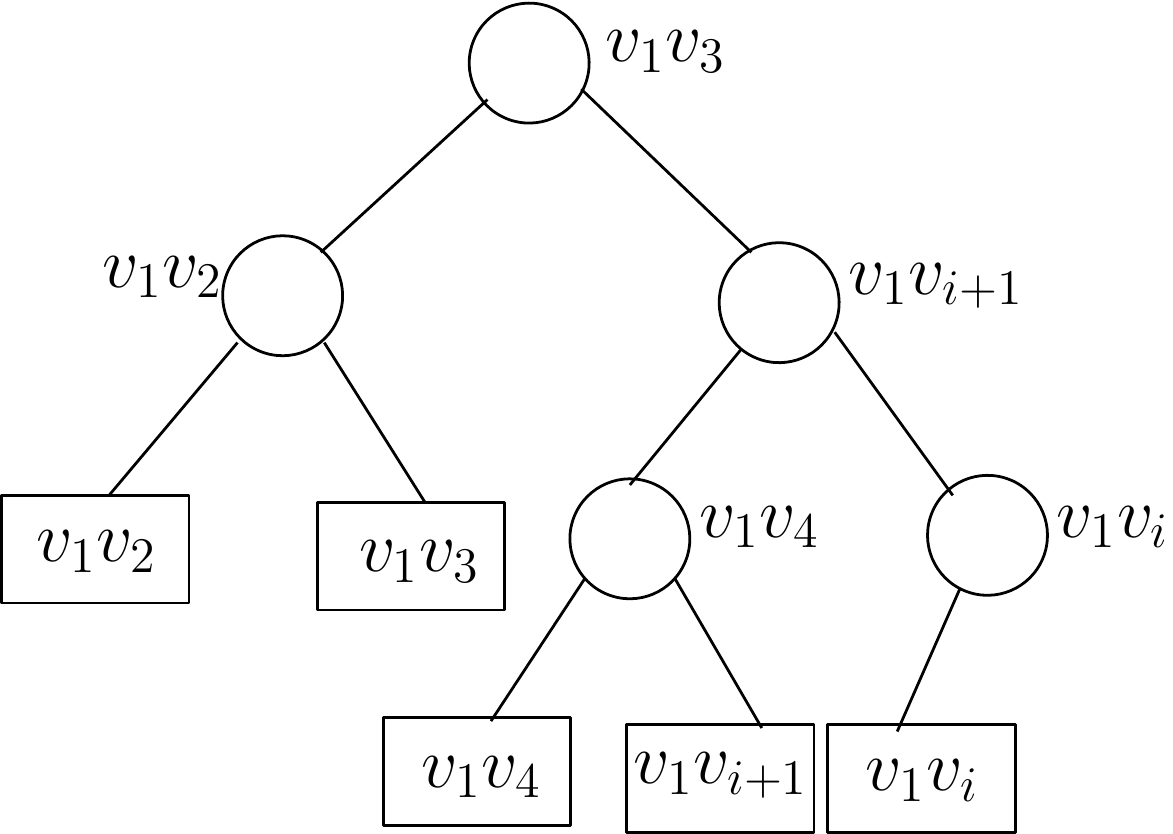}
\endminipage\hfill
\minipage{0.33\textwidth}
\includegraphics[width=\linewidth]{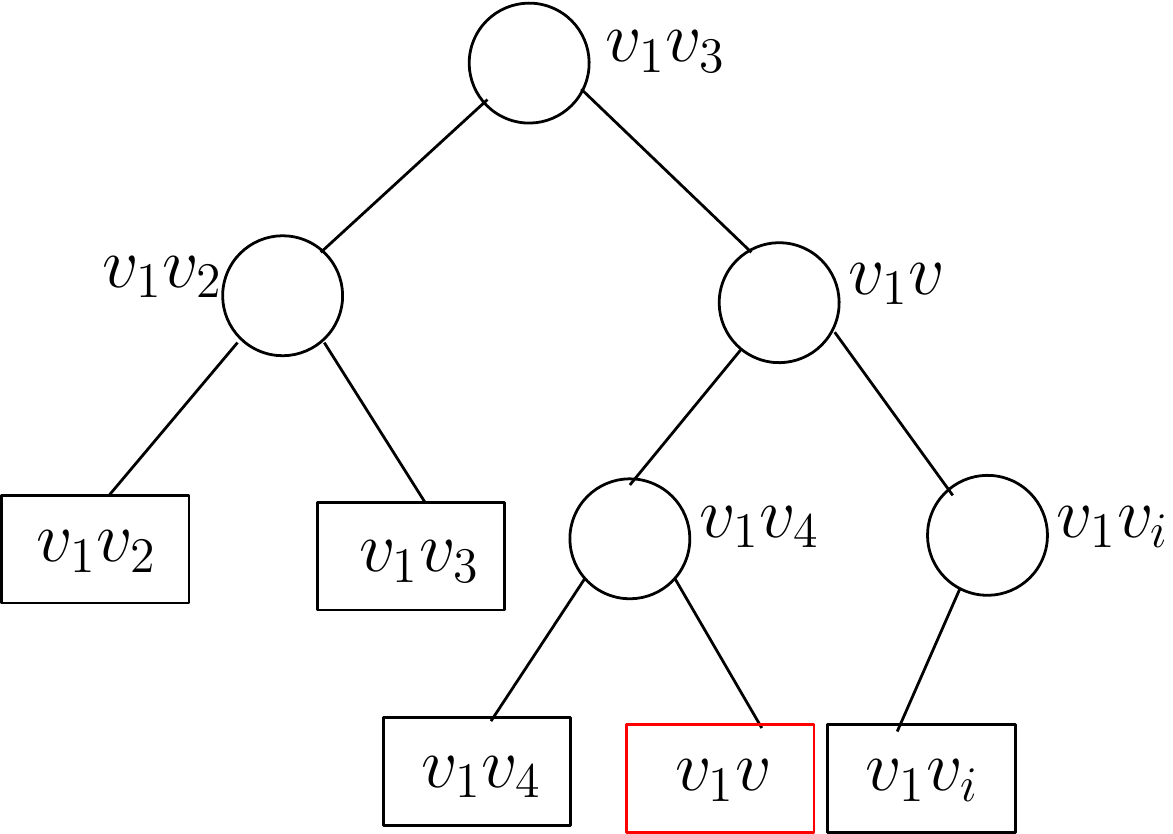}
\endminipage

\caption{Illustrating the insertion of vertex $v$ (left), the data structure $RBT_{v_1}$ before insertion of $v$ (middle), and the data structure $RBT_{v_1}$ after inserting $v$ (right).}
\label{fig:insert}
\end{figure}

\begin{lemma}
In the updated polygonal domain $\cal P''$, any edge $pq \in VG_{\cal P''}$ if and only if $p$ and $q$ are mutually visible among obstacles in $\cal P''$.
\end{lemma}
\begin{proof}
All the visible edges incident to $v$ are computed correctly using the algorithm from \cite{journals/siamcomp/KapoorM00}.
For every vertex $v'$ that is an endpoint of a visible edge incident to $v$, we do binary search in $RBT_{v'}$ to prune visible edges incident to $v'$ that intersect with triangle $vv_iv_{i+1}$ using the cone $vc_m$.
Following Lemma~\ref{lem:dynvisgr}, these edges are computed iteratively.
It ensures that an edge is removed from $VG_{\cal P'}$ only if it intersects the triangle $vv_iv_{i+1}$.
Further, by the correctness of Lemma~\ref{lem:dynvisgr}, it is guaranteed that every visible edge that is intersecting the triangle $vv_iv_{i+1}$ is removed from $VG_{\cal P'}$.
\end{proof}

\begin{lemma}
Whenever a vertex $v$ is inserted to $\cal P'$, our algorithm updates the visibility graph $VG_{\cal P'}$ in $O(k(\lg n')^2+h)$ time. 
Here, $k$ is the number of combinatorial changes in $VG_{\cal P'}$ due to the insertion of $v$, $n'$ is the number of vertices of $\cal P'$, and $h$ is the number of obstacles in $\cal P'$. 
\end{lemma}
\begin{proof}
From Theorem~\ref{thm:vpdynobst}, finding all the vertices of $\cal P'$ that are visible to $v$ takes $O(k_1(\lg{n'})^2 + h)$ time.
Here, $k_1$ is the number of vertices in $\cal P'$ that are visible to $v$.
Inserting visible edges corresponding to these vertices in $VG_{\cal P'}$ takes $O(k_1(\lg{n'}))$ time.
To recursively determining edges in $S$, search in red-black trees stored at the endpoints of edges in $S$, and to delete edges in $S$ together takes $O(k_2(\lg{n'}))$ time, where $k_2 = k-k_1$.
\end{proof}

\subsection{Deletion of a vertex} 
\label{visgraph:delete}

Let $v$ be a vertex in $\cal P'$.
Also, let $v_i$ and $v_{i+1}$ be the vertices adjacent to $v$ along the boundary of an obstacle in $\cal P'$.
Further, let $VG_{\cal P'}$ be the visibility graph of $\cal P'$.
Suppose the vertex $v$ is deleted.
Then, the deletion algorithm updates $VG_{\cal P'}$.
For line segment joining any two vertices $v', v''$ in $\cal P''$, if $v'v''$ intersects the triangle $vv_iv_{i+1}$, edge $v'v''$ needs to be introduced into $VG_{\cal P'}$.
In this algorithm, like in the insertion algorithm, we use Lemma~\ref{lem:dynvisgr} to efficiently determine the visible edges in $VG_{\cal P''}$ that does not belong to $VG_{\cal P'}$.

Let $S_1$ be the set of vertices of $\cal P'$ visible to $v$.
For each vertex $v'\in S_1$, we determine all the vertices $v''$ such that edge $v'v''$ intersects triangle $vv_iv_{i+1}$.
Without loss of generality, suppose ray $v'v_i$ makes a larger angle with ray $v'v$ as compared to ray $v'v_{i+1}$.
Let $vc_m$ be the cone with rays $v'v$ and $v'v_i$ bounding it.
The visible edges lying in cone $vc_m$ are the potential candidates to be inserted to $VG_{\cal P'}$.
We invoke the query algorithm in Section~\ref{sect:insertdelete} with query point $v'$ and the visibility cone $vc_m$. 
As a result, our algorithm finds every vertex $v''\in \cal P''$ that is visible to $v'$ such that the edge $v'v''$ resides in cone $vc_m$. 
The edge $v'v''$ is inserted as a new visible edge into both $RBT_{v'}$ and $RBT_{v''}$.
In specific, for every $j \ge 1$, let $S_{j+1}$ be equal to $\bigcup_{v' \in S_j} S_{v'}$.
Following the same procedure, for every $j > 1$, for every vertex $v' \in S_j$ and $v' \notin S_{j'}$ with $j' < j$, we find the set $S_{j+1}$ comprising of all the visible edges in $\cal P''$ that are incident to $v'$ and intersect the triangle $vv_iv_{i+1}$.
If at any point, $S_j \subseteq S_{j'}$ for some $j'< j$, we terminate the algorithm.
In the end, we delete $RBT_v$ as part of removing all the visible edges of $VG_{\cal P'}$ that are incident to $v$.
(Refer to Fig.~\ref{fig:delete}.)

\begin{figure}[h]
\centering

\minipage{0.29\textwidth}
\includegraphics[width=\linewidth]{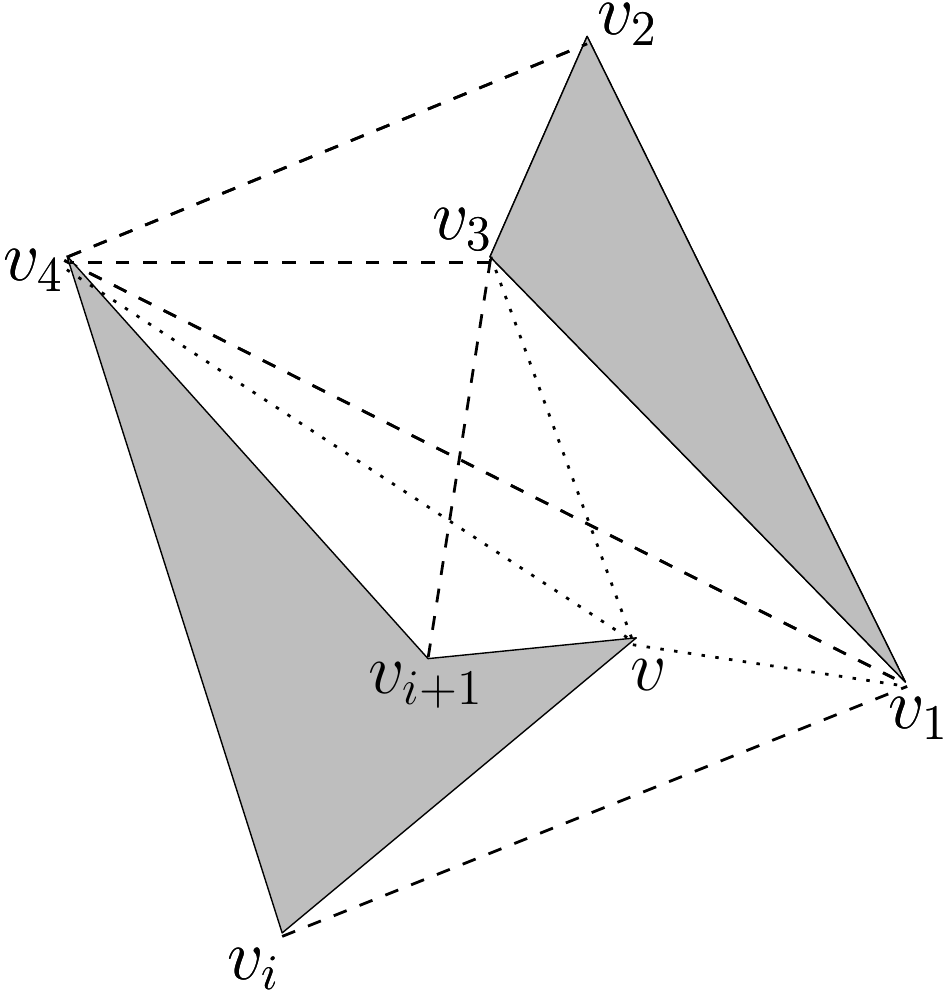}
\endminipage\hfill
\minipage{0.33\textwidth}
\includegraphics[width=0.95\linewidth]{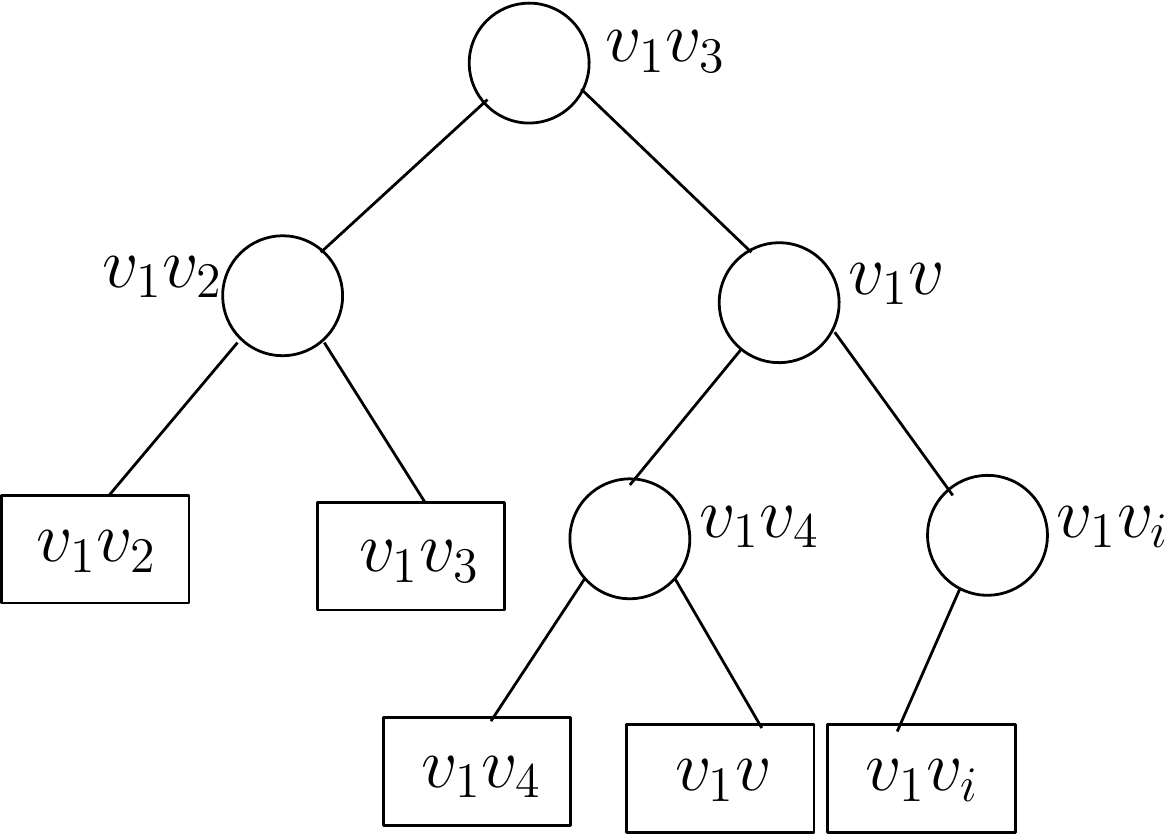}
\endminipage\hfil
\minipage{0.33\textwidth}
\includegraphics[width=\linewidth]{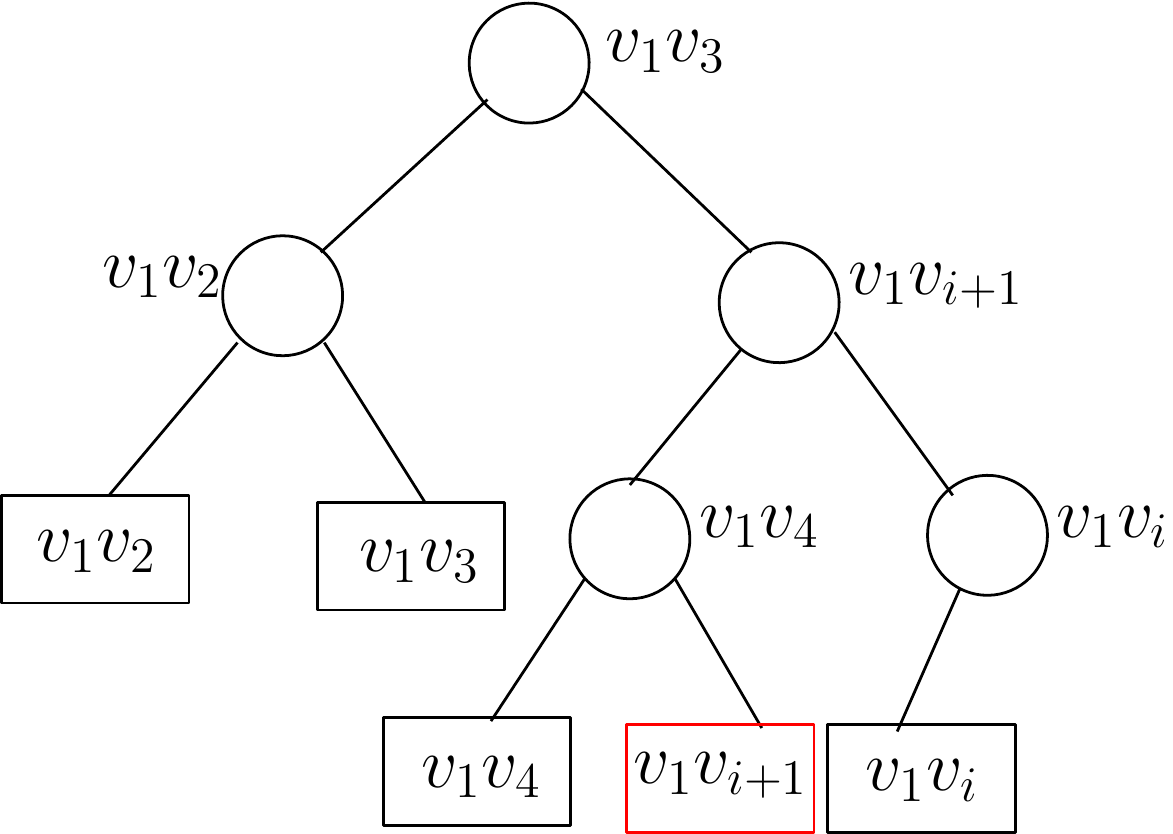}
\endminipage\hfill
%\minipage{0.33\textwidth}
%\includegraphics[width=\linewidth]{afterinsert3.pdf}
%\endminipage

\caption{Illustrating the vertex $v$ to be deleted (left), the data structure $RBT_{v_1}$ before deleting $v$ (middle), and the data structure $RBT_{v_1}$ after deleting $v$ (right).}
\label{fig:delete}
\end{figure}

\begin{lemma}
In any updated polygonal domain $\cal P''$, any edge $pq \in VG_{\cal P''}$ if and only if $p$ and $q$ are mutually visible among obstacles in $\cal P''$.
\end{lemma}
\begin{proof}
Let a vertex $v$ is deleted from the boundary of an obstacle in $\cal P'$.
For every vertex $v'\in \cal P'$ that is visible to $v$, the edge $vv'$ is deleted from $VG_{\cal P'}$.
No other edge needs to be deleted from the current visibility graph.
For any two vertices $v'v''$, if the edge intersect the triangle $vv_iv_{i+1}$, it should be added as a visible edge in $VG_{\cal P'}$.
Following Lemma~\ref{lem:dynvisgr}, these edges are computed iteratively.
In each iteration, for every vertex $v'$, we invoke the query algorithm from Section~\ref{sect:queryalgo} with $v'$ as the query point and $vc_m$ as the visibility cone.
By the correctness of Lemma~\ref{lem:dynvisgr} and Lemma~\ref{lem:querycorr}, it is guaranteed that every visible edge that is intersecting the triangle $vv_iv_{i+1}$ is determined correctly and added to $VG_{\cal P'}$.
\end{proof}

\begin{lemma} 
Whenever a vertex $v$ is deleted from $\cal P'$, our algorithm updates the visibility graph $VG_{\cal P'}$ in $O(k((\lg n')^2+h))$ time. 
Here, $k$ is the number of combinatorial changes in $VG_{P'}$ due to the deletion of $v$, $n'$ is the number of vertices of $\cal P'$, and $h$ is the number of obstacles in $\cal P'$. 
\end{lemma}
\begin{proof}
All the visible edges in $VG_{\cal P'}$ that are incident to $v$ are deleted in $O(k_1 \lg{n'})$ time.
Here, $k_1$ is the number of visible edges incident to $v$.
From Theorem~\ref{thm:vpdynobst}, all the invocations of the visibility polygon query algorithm together takes $O(k_2((\lg n')^2+h))$ time, where $k_2 = k-k_1$ is the number of visible edges in $VG_{\cal P''}$ that are determined with this algorithm.
Inserting the newly found visible edges into corresponding red-black trees takes $O(k_2 \lg{n'})$ time.
\end{proof}

\begin{theorem}
Given a polygonal domain $\cal P$ defined with $h$ obstacles and $n$ vertices, we preprocess $\cal P$ in $O(n(\lg{n})^2+h(\lg h)^{1+\epsilon}+|E|\lg |E|)$ time to construct data structures of size $O(n+|E|)$ so that whenever a vertex $v$ is inserted to (resp. deleted from) the current polygonal domain $\cal P'$, the algorithm updates the data structures that save visibility graph of $\cal P'$ in $O(k(\lg n')^2+h)$ (resp. $O(k((\lg n')^2+h))$) time.
Here, $\epsilon$ is a small positive constant resulting from the triangulation of the free space $\cal{F(P)}$ using the algorithm in \cite{journals/ijcga/Bar-YehudaC94}, $|E|$ is the number of edges present in the visibility graph of $\cal P$, $k$ is the number of combinatorial changes in updating the visibility graph of $\cal P'$ due to the insertion (resp. deletion) of $v$, and $n'$ is the number of vertices of $\cal P'$.
\end{theorem}

\section{Conclusions}
\label{sect:conclu}
We proposed algorithms for the following problems among dynamic polygonal obstacles in the plane: (i) dynamically updating the visibility polygon of any given query point, (ii) answering visibility polygon queries, and (iii) dynamically maintaining the visibility graph of the polygonal domain.
An immediate extension of this algorithm would handle new obstacles being introduced into the polygonal domain as well as allowing the newly inserted vertex to be in a corridor different from the corridor in which its neighbors reside.
Further, it would be interesting to devise efficient dynamic algorithms for these problems so that the update time complexities mainly depend on the number the combinatorial changes to update the visibility polygon.
In specific, as other parameters are under logarithm, this involves removing the dependency on the number of obstacles in the polygonal domain from the update time complexities.

\subsection*{Acknowledgement}

This research of R. Inkulu is supported in part by SERB MATRICS grant MTR/2017/000474.
The authors also like to acknowledge reviewers of the preliminary version of this paper for their valuable input.

\bibliographystyle{plain}
%\bibliography{../ajar/results/bibs/weiregsp,../ajar/results/bibs/geomgraphs,../ajar/results/bibs/misc,../ajar/results/bibs/shortestpaths,../ajar/results/bibs/visibility,../ajar/results/bibs/nearneigh,../ajar/results/bibs/voronoi}

\end{document}